\documentclass[12pt]{article}
\usepackage{hyperref}
\usepackage{graphicx}
\usepackage{curves}
\usepackage{amssymb,amsmath,xcolor,color,framed,epstopdf,amsthm}
\usepackage[english]{babel}
\usepackage{enumerate}
\usepackage{mathrsfs}
\usepackage{color}
\usepackage{empheq}
\usepackage{braket}
\usepackage{hyperref}
\usepackage{cancel}
\usepackage{subcaption}
 \usepackage{lipsum}
 \usepackage{mathtools}
 \usepackage{MnSymbol}
\usepackage{xcolor}
\usepackage{relsize}
\usepackage{tikz}
\usetikzlibrary{shapes,snakes}
\usepackage{pgf}
\usepackage{pgflibraryarrows}
\usepackage{pgflibrarysnakes}
\usetikzlibrary{decorations.text}
\usepgfmodule{shapes}
\usetikzlibrary{decorations.pathmorphing}
\usetikzlibrary{decorations.markings}
\usetikzlibrary{patterns}
\usetikzlibrary{automata}
\usetikzlibrary{positioning}
\usepackage{pgfplots}
\usetikzlibrary{backgrounds}
\tikzset{partial ellipse/.style args={#1:#2:#3}{insert path={+ (#1:#3) arc (#1:#2:#3)} }}
\tikzset{->-/.style={decoration={ markings, mark=at position #1 with {\arrow{>}}},postaction={decorate}}}
\newtheorem{theorem}{Theorem}[section]
\newtheorem{lemma}[theorem]{Lemma}

\newtheorem{remark}[theorem]{Remark}

\newtheorem{assumption}{Assumption}[section]

\newtheorem{prop}[theorem]{Proposition}
\newtheorem{thm}[theorem]{Theorem}
\newtheorem{RHP}{Riemann--Hilbert problem}
\theoremstyle{definition}
\newcommand{\R}{\mathbb{R}}

\newcommand{\C}{\mathbb{C}}

\newcommand{\be}{\begin{equation}}
\newcommand{\ee}{\end{equation}}
\newcommand{\bea}{\begin{eqnarray}}
\newcommand{\eea}{\end{eqnarray}}

\def\XXint#1#2#3{{\setbox0=\hbox{$#1{#2#3}{\int}$}
     \vcenter{\hbox{$#2#3$}}\kern-.5\wd0}}

\def\d{{\rm d}}

\def\i{{\rm i}}
\def\1{\operatorname{Id}}

\def\Res{\mathop{{\rm Res}}}

\def\Re{\operatorname{Re}}
\def\Im{\operatorname{Im}}
\def\exp{\operatorname{exp}}
\def\Bes{\mathrm{Bes}}

\def\le{\left}
\def\ri{\right}

\numberwithin{equation}{section}

\title{ \LARGE\bf Large-space and large-time asymptotics for the focusing nonlinear Schr\"{o}dinger soliton gas}
\author{\hspace{0.6 cm}{}Dedi Yan$^{a}$, Xianguo Geng$^{a,b}$, Jiao Wei$^{a}$\footnote{\footnotesize
		Corresponding author.{\sl E-mail address}:math.jwei@zzu.edu.cn}\\
	\leftline{\hspace{0.6 cm}{\small{\sl $^{a}$ School of Mathematics and Statistics, Zhengzhou University, 100 Kexue Road, Zhengzhou, }}}\\
	\leftline{\hspace{0.6 cm}{\small{\sl \quad Henan 450001, People's Republic of China}}}\\
	\leftline{\hspace{0.6 cm}{\small{\sl $^{b}$ School of Mathematics and Statistics, North China University of Water Resources }}}\\
	\leftline{\hspace{0.6 cm}{\small{\sl \quad and Electric Power, Zhengzhou, Henan 450011, People's Republic of China}}}}
\date{}

\begin{document}

\maketitle
\begin{abstract}
We investigate the large-space and large-time asymptotic behavior of a soliton gas for the focusing nonlinear Schr\"odinger equation. The soliton gas is constructed as the continuum limit of pure $N$-soliton solutions when $N\to\infty$, with the discrete spectrum confined to two segments $\Sigma_1$ and $\Sigma_2$. In particular, our framework does not require the discrete spectrum to be confined to the imaginary axis. By combining the nonlinear steepest descent method with an appropriate $g$-function mechanism, we show that, as $x\to-\infty$, the soliton gas is asymptotically described by a finite-gap elliptic solution with constant coefficients. In the large-time regime $t\to+\infty$, we assume that the endpoint $F$ lies on the trajectory of $H(\xi)$ with $\xi=\frac{x}{2t}\in(-E_1-\sqrt{2}E_2,-E_1)$, namely, $F=H(\hat{\xi})$, $\hat{\xi}\in (-E_1-\sqrt{2}E_2,-E_1)$. Under this assumption, we prove that the solution exhibits distinct asymptotic behaviors in different regions of the variable $\xi=\frac{x}{2t}$. More precisely, there exist an exponentially decaying region $\xi\in(-E_1,+\infty)$, a modulated elliptic-wave region $\xi\in(\hat{\xi},-E_1)$, and an unmodulated elliptic-wave region $\xi\in(-\infty,\hat{\xi})$.\\

\end{abstract}

\section{Introduction}

The concept of a soliton gas, understood as an infinite statistical ensemble of interacting solitons, was first introduced by Zakharov for the Korteweg--de Vries (KdV) equation \cite{Zak71}. It was later extended to the focusing nonlinear Schr\"odinger (NLS) equation by El and Tovbis \cite{E4,TW2022}. Soliton gas is closely connected with integrable nonlinear partial differential equations, such as the KdV equation, the Camassa--Holm equation, and the focusing NLS equation, all of which are amenable to the inverse scattering transform \cite{GGKM}.

In recent years, the asymptotic analysis of soliton gases has attracted considerable attention. Girotti \emph{et al.} carried out a detailed study of the large-space and large-time asymptotics of the KdV soliton gas in \cite{Girotti-1}, based on the RH formulation of the corresponding $N$-soliton solutions. They further investigated the behavior of a dense mKdV soliton gas and its large-time dynamics in the presence of a single trial soliton \cite{Girotti-2}. This approach was motivated by the notion of primitive potentials introduced by Dyachenko, Zakharov, and collaborators \cite{DZZ}. Then Wang \emph{et al.} analyzed the genus two KdV soliton gas \cite{Wang}. The large-space and large-time asymptotics of soliton gas solutions for the Camassa--Holm equation were studied in \cite{GYJ}. Zhang \emph{et al.} analyzed the large-$x$ asymptotics of the focusing NLS soliton gas with discrete spectrum confined to the imaginary axis, as well as the mKdV soliton gas under nonzero boundary conditions \cite{HXF,ZXE2}. Related phenomena of soliton and breather shielding for the focusing NLS equation were investigated by Bertola and collaborators in \cite{Grava-3,NLSGAS,FG}. The work of Gkogkou, Mazzuca and McLaughlin provides a rigorous description of the formation of a soliton-gas condensate from \(N\)-soliton solutions in the condensate limit \(N\to\infty\), with \((x,t)\) restricted to compact sets~\cite{GMM}.
 Orsatti’s thesis studies focusing NLS soliton gases generated by spectral points filling planar domains and, in the elliptic case, reduces the problem to a line-segment RH problem whose foci are assumed to lie on the imaginary axis~\cite{Orsatti}. Our work is complementary in that we start directly from a pure-soliton RH problem with spectral points condensing on two general conjugate segments in the complex plane. This allows the limiting soliton-gas RH problem to include spectral arcs not confined to the imaginary axis. Under a solvability condition for the associated g-function system, we obtain the large-space and large-time asymptotics, including an exponentially decaying region, a modulated elliptic-wave region, and an unmodulated elliptic-wave region.

In this paper, we study the large-space and large-time asymptotics of a soliton gas for the focusing NLS equation
\begin{align}\label{NLS}
	iq_{t}+\frac{1}{2}q_{xx}+|q|^2q=0,\qquad -\infty<x<+\infty,\quad t>0,
\end{align}
which is one of the fundamental models in mathematical physics and integrable systems. The focusing NLS equation has been extensively studied because of both its rich mathematical structure and its broad physical relevance. It arises, for example, as a canonical model for pulse propagation in optical fibers and for one-dimensional wave propagation in deep water. The semiclassical, or zero-dispersion, limit of the focusing NLS equation has been studied in \cite{TA3,TA1,TA2}. Buckingham and Venakides characterized the long-time asymptotics for a shock-type problem of the focusing NLS equation by means of the nonlinear steepest descent method \cite{BR,DeiftZ,CMP,GL}. Long-time asymptotics for the focusing NLS equation with nonzero boundary conditions and the nonlinear stage of modulational instability were analyzed by Biondini \emph{et al.} in \cite{BGL,BG}. Bilman and Miller developed a robust inverse scattering transform for the focusing NLS equation with nonzero boundary conditions and used it to place the Peregrine solution and related higher-order rogue waves into an inverse-scattering framework \cite{BD}. Boutet de Monvel and collaborators studied the long-time asymptotics of the focusing NLS equation with step-like oscillatory background in a series of works \cite{Boutet1,Boutet2,Boutet3,Boutet4,Boutet5}.

From a physical point of view, a soliton gas may be regarded as a dense nonlinear wave field formed by a large number of bright solitons. While a single bright soliton represents a localized coherent wave packet, a soliton gas arises when the number of solitons becomes large and the associated discrete spectrum condenses onto continuous spectral arcs. In this limit, the resulting wave field behaves as a nonlinear statistical ensemble of interacting solitons. The main purpose of the present paper is to describe the macroscopic organization of such a wave field in the large-space and large-time regimes. In particular, our asymptotic analysis shows how the focusing NLS soliton gas separates into distinct self-similar regions, including an exponentially small region, a modulated elliptic-wave region, and an unmodulated elliptic-wave region.

For integrable initial-value problems such as \eqref{NLS}, the inverse scattering transform reduces the nonlinear evolution to a matrix Riemann--Hilbert (RH) problem in the spectral variable. In the present work, we begin with a pure-soliton RH problem whose discrete spectrum is distributed along the two segments $\Sigma_1\cup\Sigma_2$; see Fig.~\ref{contour1}. In particular, our framework does not require the discrete spectrum to be confined to the imaginary axis. By taking the limit as the number of solitons tends to infinity, we arrive at an RH problem describing the associated soliton gas. We then employ the nonlinear steepest descent method, together with an appropriate $g$-function mechanism \cite{DeiftItsZhou,DeiftVZ,DeiftV,Fan,Yan}, to derive explicit large-time asymptotics for the focusing NLS soliton gas.

A central step in the construction of the \(g\)-function is to prove the solvability of the nonlinear system that determines its parameters. The corresponding asymptotic region is present only when this system admits a solution. Starting from an arbitrary initial point \(E\), the system \eqref{system21}--\eqref{system23} generates an \(E\)-dependent curve, denoted by \(\gamma(E)\); see, for example, Fig.~\ref{trajectory}. We require the endpoint \(F\) to lie on this curve.
More precisely, in the large-time analysis, \(F\) is assumed to belong to the trajectory traced out by \(H(\xi)\), where \(H(\xi)\) denotes the solution of \eqref{system21}--\eqref{system23}; see Lemma~\ref{gexist}. Thus, for
$
\xi=\frac{x}{2t}\in(-E_1-\sqrt{2}E_2,-E_1),
$
we impose the condition
\[
F=H(\hat{\xi}),\qquad
\hat{\xi}\in(-E_1-\sqrt{2}E_2,-E_1).
\]
This condition ensures the existence of the \(g\)-function required for the subsequent nonlinear steepest descent analysis.

We then replace the original phase function $\theta(k)$ by an analytic function $g(k)$ chosen so that, after suitable triangular factorizations and contour deformations, the exponentially growing jumps are transformed into constant jumps independent of $k$, while the remaining jump matrices approach the identity exponentially fast. To justify the signature table of the introduced $g$-function, we establish Lemma \ref{lem:Im-gF-monotone}, which shows that
$$
\Im \tilde{g}(F;\xi)>0,\qquad \xi\in(\hat{\xi},-E_1).
$$
This property ensures that the associated model RH problem can be solved explicitly in terms of Riemann theta functions and Abelian integrals on a suitable Riemann surface. As a consequence, one obtains a modulated elliptic-wave region in this regime.

Our main results can be summarized as follows.

\begin{thm}
As $x\to+\infty$, assume that $E,F, r(k)$ satisfies the Assumption \ref{assumption1}, \ref{assumption:nonintersecting-spectrum}, the soliton gas $q(x,0)$ decays exponentially to zero. As $x\to-\infty$, it admits the asymptotic representation
\begin{gather}\label{qx}
	\begin{array}{l}
		q(x,0)=-(\Im E+\Im F)f^2_{\infty}e^{2ixg_{\infty}}
		\frac{\vartheta_3\le(A(\infty)+A(E_0)+\frac{\tau}{2}+\frac{x\Omega+\Delta}{2\pi};\tau\ri)}
		{\vartheta_3\le(A(\infty)+A(E_0)+\frac{1}{2}+\frac{\tau}{2};\tau\ri)}
		\frac{\vartheta_3\le(A(\infty)-A(E_0)-\frac{1}{2}-\frac{\tau}{2};\tau\ri)}
		{\vartheta_3\le(A(\infty)-A(E_0)-\frac{\tau}{2}-\frac{x\Omega+\Delta}{2\pi};\tau\ri)}
		+\mathcal{O}\!\left(\frac{1}{x}\right),
	\end{array}
\end{gather}
where $g_{\infty}$, $\Omega$, $\Delta$, $f_{\infty}$, and $A(k)$ are defined in \eqref{ginfty}, \eqref{Omega}, \eqref{Delta}, \eqref{finfty}, and \eqref{Ak}, respectively.
\end{thm}

\begin{theorem}
Assume that $E,F, r(k)$ satisfies the Assumption \ref{assumption1}, \ref{assumption:nonintersecting-spectrum}  and $F$ lies on the trajectory of the function $H(\xi)$, where
\[
\xi=\frac{x}{2t}\in(-E_1-\sqrt{2}E_2,-E_1),\qquad E_1=\Re E,\quad E_2=\Im E.
\]
Then, in the large-time regime, the focusing NLS soliton gas exhibits the following asymptotic behavior:
\begin{enumerate}
	\item
	For $\xi$ in any compact subset of $(-E_1,+\infty)$, one has uniformly
	\[
	q(x,t)=\mathcal{O}(e^{-ct}),
	\]
	where $c>0$ is a constant.
	
	\item
	For $\xi$ in any compact subset of $(\hat{\xi},-E_1)$, where $H(\hat{\xi})=F$ and
	\[
	-E_1-\sqrt{2}E_2<\hat{\xi}<-E_1,
	\]
	the solution satisfies
	\begin{gather}\label{qxt1}
		\begin{array}{l}
			q(x,t)=-(\Im E+\Im H)\tilde{f}^2_{\infty}e^{2it\tilde{g}_{\infty}}
			\frac{\vartheta_3\le(\tilde{A}(\infty)+\tilde{A}(\tilde{E}_0)+\frac{\tilde{\tau}}{2}+\frac{t\tilde{\Omega}+\tilde{\Delta}}{2\pi};\tilde{\tau}\ri)}
			{\vartheta_3\le(\tilde{A}(\infty)+\tilde{A}(\tilde{E}_0)+\frac{1}{2}+\frac{\tilde{\tau}}{2};\tilde{\tau}\ri)}
			\frac{\vartheta_3\le(\tilde{A}(\infty)-\tilde{A}(\tilde{E}_0)-\frac{1}{2}-\frac{\tilde{\tau}}{2};\tilde{\tau}\ri)}
			{\vartheta_3\le(\tilde{A}(\infty)-\tilde{A}(\tilde{E}_0)-\frac{\tilde{\tau}}{2}-\frac{t\tilde{\Omega}+\tilde{\Delta}}{2\pi};\tilde{\tau}\ri)}
			+\mathcal{O}\!\left(\frac{1}{t}\right),
		\end{array}
	\end{gather}
	where $\tilde{f}_{\infty}$, $\tilde{g}_{\infty}$, $\tilde{A}(k)$, $\tilde{\Omega}$, and $\tilde{\Delta}$ are defined in \eqref{tfinfty}, \eqref{tginfty}, \eqref{tAk}, \eqref{tOmega}, and \eqref{tDelta}, respectively.
	
	\item
	For $\xi$ in any compact subset of $(-\infty,\hat{\xi})$, one has uniformly
	\begin{gather}\label{qxt2}
		\begin{array}{l}
			q(x,t)=-(\Im E+\Im F){f}^2_{\infty}e^{2it\hat{g}_{\infty}}
			\frac{\vartheta_3\le({A}(\infty)+{A}({E}_0)+\frac{{\tau}}{2}+\frac{t\hat{\Omega}+{\Delta}}{2\pi};{\tau}\ri)}
			{\vartheta_3\le({A}(\infty)+{A}({E}_0)+\frac{1}{2}+\frac{{\tau}}{2};{\tau}\ri)}
			\frac{\vartheta_3\le({A}(\infty)-{A}({E}_0)-\frac{1}{2}-\frac{{\tau}}{2};{\tau}\ri)}
			{\vartheta_3\le({A}(\infty)-{A}({E}_0)-\frac{{\tau}}{2}-\frac{t\hat{\Omega}+{\Delta}}{2\pi};{\tau}\ri)}
			+\mathcal{O}\!\left(\frac{1}{t}\right),
		\end{array}
	\end{gather}
	where $f_{\infty}$, $\hat{g}_{\infty}$, $A(k)$, and $\Delta$ are defined in \eqref{finfty}, \eqref{hginfty}, \eqref{Ak}, and \eqref{Delta}, respectively.
\end{enumerate}
\end{theorem}

\begin{remark}
The error estimates in the three long-time regions are uniform only on compact subsets of the corresponding open $\xi$-intervals.  No uniformity is claimed as $\xi$ approaches $-E_1$ or $\hat\xi$.  At these transition values the local geometry of the phase changes and a separate double-scaling analysis is required.
\end{remark}

The paper is organized as follows. In Section \ref{sec:2}, we begin with the pure $N$-soliton solutions of the focusing NLS equation \eqref{NLS}, whose spectrum is confined to the segments $\Sigma_1\cup\Sigma_2$, and then derive the RH problem for the corresponding soliton gas by letting $N\to\infty$. In Section \ref{sec:3}, we show that, as $x\to-\infty$, the soliton gas is asymptotically described by a finite-gap elliptic function up to an error of order $\mathcal{O}(x^{-1})$, whereas as $x\to+\infty$ it decays exponentially to zero. In Sections \ref{sec:4}--\ref{sec:6}, we establish the global long-time asymptotics of $q(x,t)$ by means of the nonlinear steepest descent method and the $g$-function mechanism. More precisely, the solution decays exponentially in the region $\xi>-E_1$, takes the form of a modulated elliptic wave in the region $\hat{\xi}<\xi<-E_1$, and is asymptotically described by a finite-gap elliptic function with unmodulated coefficients in the region $\xi<\hat{\xi}$.

\section{Soliton gas as the limit of $N$-solitons as $N\to +\infty$}
\label{sec:2}
The focusing NLS equation \eqref{NLS} is a prime example of a completely integrable system and the  Lax pair of the focusing NLS equation reads
\begin{align}\label{Lax-O}
(\partial_{x}-\mathcal{L})\Phi=0,\quad \mathcal{L}=-ik\sigma_3+Q,\\
(i\partial_{t}-\mathcal{B})\Phi=0,\quad \mathcal{B}=ik\mathcal{L}+\frac{1}{2}\sigma_3(Q^2-Q_{x}),
\end{align}
where
\begin{align*}
Q=Q(x,t)=\begin{pmatrix}0&q(x,t)\\-q(x,t)^*&0  \end{pmatrix},
\end{align*}
the $q(x,t)^*$ denotes the complex conjugate of the complex potential function $q(x,t)$ and $\sigma_3$ is the third Pauli matrix $\sigma_1=\begin{pmatrix}0&1\\1&0  \end{pmatrix},\ \sigma_2=\begin{pmatrix}0&-i\\i&0  \end{pmatrix},\ \sigma_3=\begin{pmatrix}1&0\\0&-1  \end{pmatrix}$.

The RH problem for the solution of the focusing NLS equation is described as follows \cite{NLS2017}:
 \begin{RHP}
Find a matrix-valued function $M(k;x,t)$ with the following properties
\begin{enumerate}
\item $M(k;x,t)$ is meromorphic in $\C\backslash \R$ with simple poles at $\{\kappa_{j} \}_{j=1}^{N}$ in $\mathbb{C}_{+}$ and at the corresponding conjugate points $\{\kappa^{*}_{j} \}_{j=1}^{N}$ in $\mathbb{C}_{-}$.
\item The boundary values $M_{\pm}(k;x,t)=M(k\pm i\varepsilon;x,t)$ satisfies the following jump relation
\begin{align}M_{+}(k;x,t)=M_-(k;x,t)V(k),\quad k \in \mathbb{R},\end{align}
\begin{align}
V(k)= \begin{pmatrix}1+|R_1(k)|^2&R^*_1(k){e}^{-2it\theta(k)}\\R_1(k){e}^{2it\theta(k)}&1 \end{pmatrix},
\end{align}
where $R_1(k)$ is the reflection coefficient and the phase function $\theta(k)=k^2+2k\xi$ with $\xi=\frac{x}{2t}$.
\item  The residue conditions:
\begin{equation}\label{residue_soliton}
\Res_{k=\kappa_{j}} M (k)= \lim_{k \to \kappa_{j}} M(k) \begin{pmatrix} 0 & 0\\  -i\chi_{j}e^{2 i t\theta(k) } & 0 \end{pmatrix}, \ \  \Res_{k=\kappa^{*}_{j}} M(k) = \lim_{k \to \kappa^{*}_{j}} M(k) \begin{pmatrix} 0 & -i\chi^{*}_j e^{-2 i t\theta(k)}\\0 & 0\end{pmatrix},
\end{equation}
where $\chi_j\in \mathbb{C}\backslash\{0\}$.
\item The normalization:
$\displaystyle M(k;x,t)= \begin{pmatrix}1&0\\0&1 \end{pmatrix} + \mathcal{O}\le(\frac{1}{k}\ri)$ as $k \to \infty$.
\item  The symmetry condition:
\begin{align}
\label{re}
M^{*}(k^{*};x,t)=\sigma_{2}M(k;x,t)\sigma_{2}.
\end{align}
\end{enumerate}
\end{RHP}
In this paper we only consider the pure $N$-soliton solutions and the reflection coefficient $R_{1}(k)$ vanishes. The  potential $q(x,t)$ is determined from $M(k;x,t)$ via
\begin{align}\label{u-sol}
q(x,t)=2i\lim\limits_{k\to\infty}kM_{12}(k;x,t).
\end{align}
We are now interested in the limit as $N  \to + \infty$ under the additional assumptions:
\begin{assumption}\label{assumption1}
\begin{enumerate}
\item Set a segment $\Sigma_{1}$ connecting $E$ and $F$, and a segment $\Sigma_{2}$ connecting $F^*$ and $E^*$, where $E=E_1+E_2i,\ F=F_1+F_2i$,  $E_1=\Re E<0,\ E_2=\Im E>0, F_1=\Re F<0,\ F_2=\Im F>0 $ and $E_1<F_1$. As the number of poles increases to infinity, the $N$ poles are equally spaced along
$\Sigma_{1}$ with distance between two poles equal to $|\Delta k|=\frac{\sqrt{(E_1-F_1)^2+(E_2-F_2)^2}}{N}$, the oriented contours $\Sigma_{1}$ and $\Sigma_{2}$ see Fig. \ref{contour1}.
\item The coefficients $\{\chi_j\}_{j=1}^N$ are assumed to be a discretization of a given function:
\begin{gather}
\chi_j = -i\frac{r(\kappa_j)(F-E)}{N\pi}, \quad j=1,\ldots, N,
\end{gather}
where $ r(k)$ is an analytic and non-vanishing in the open set containing all deformation contours and lens regions.
\end{enumerate}
\end{assumption}
 The large-$x$ behavior of the soliton-gas potential will  be derived directly from the limiting gas RH problem below, after the continuum limit has been established.
Drawing on the strategy employed in \cite{Girotti-1}, we  derive the meromorphic RH problem of the focusing NLS soliton gas as the limit $N \to+\infty$ of $N$-solitons. For convenience, we introduce a closed curve $\Gamma_+$ located in the upper half-plane $\mathbb{C}{+}$, encircling the poles $\kappa_j$ in a counterclockwise direction. Similarly, $\Gamma_-$ is a clockwise contour that surrounds the poles $\kappa^*_j$ in the lower half-plane $\mathbb{C}_{-}$. We first remove the poles by defining
\begin{equation}
Z(k;x,t) =M(k;x,t)
\begin{cases}\begin{pmatrix}1&0\\ \displaystyle\sum\limits_{j=1}^{N}\frac{i\chi_j e^{2 i t\theta(k)}}{k -\kappa_{j}}&1\end{pmatrix},
      &\text{$k$ inside $\Gamma_+$},\\
 \begin{pmatrix}1& \displaystyle\sum\limits_{j=1}^{N}\frac{i\chi^*_j e^{-2it\theta(k)}}{k-\kappa^*_{j}}\\ 0 &1 \end{pmatrix},
      &\text{$k$ inside $\Gamma_-$},\\
    I,
    &\text{otherwise},
\end{cases}
\end{equation}
where $I$ is the $2\times 2$ identity matrix. Then the matrix-valued function $Z(k;x,t) $ satisfies the following jump relation
\begin{gather}\label{Zjump-discrete}
Z_+(k) = Z_-(k) \begin{cases} \begin{pmatrix} 1 & 0\\ \displaystyle \sum_{j=1}^{N}\frac{i\chi_j e^{2 it\theta(k)}}{k-\kappa_{j}} & 1\end{pmatrix}, & k \in \Gamma_+, \\
 \begin{pmatrix}1& -\displaystyle {\sum_{j=1}^{N}\frac{i\chi^*_j e^{-2it\theta(k)}}{k -\kappa^*_{j}}}\\ 0 &{1} \end{pmatrix}, & k \in \Gamma_-,
\end{cases}
\end{gather}
where, for $k\in \Gamma_+ \cup \Gamma_- $, the boundary values $Z_{+}(k)$ are taken from the left side of the contours, and the boundary values $Z_{-}(k)$ are taken from the right.
Then the following proposition holds.
\begin{prop}\label{p1}
For any open set $K_{+}$ containing the segments $\Sigma_{1}$, and any open set  $K_{-}$ containing the segment $\Sigma_{2}$,  the following limit holds uniformly for all  $k \in \mathbb{C }\backslash K_{+}$:
\begin{gather}
\lim_{N\to +\infty} \sum_{j=1}^{N}\frac{i\chi_j }{k -\kappa_{j}} =\int_{E}^{F} \frac{2ir(\zeta)}{k - \zeta}\frac{\d \zeta}{2\pi i} \,,
\end{gather}
and the following limit holds uniformly for all $k \in \mathbb{C} \backslash K_{-}$:
\begin{gather}
\lim_{N\to +\infty}\sum_{j=1}^{N}\frac{i\chi^*_j }{k-\kappa^*_{j}} = \int_{F^*}^{E^*} \frac{2i  r^*(\zeta^*)}{k - \zeta}\frac{\d \zeta}{2\pi i}.
\end{gather}
\end{prop}
\begin{proof}
Parametrize the oriented segment $\Sigma_1$ by $\zeta(s)=E+s(F-E)$, $0\leq s\leq1$, and choose the equally spaced points $\kappa_j=E+s_j(F-E)$ with $s_j=j/N$ (a different endpoint convention changes the sum by $\mathcal O(N^{-1})$ only).  By Assumption~\ref{assumption1},
\[
 i\chi_j=\frac{F-E}{N\pi}r(\kappa_j),
\]
and hence, uniformly for $k$ in a compact subset of $\mathbb C\setminus\Sigma_1$,
\begin{align}
 \sum_{j=1}^N\frac{i\chi_j}{k-\kappa_j}
 &=\frac{F-E}{\pi}\frac1N\sum_{j=1}^N
 \frac{r(E+s_j(F-E))}{k-E-s_j(F-E)}\nonumber\\
 &\longrightarrow\frac{F-E}{\pi}\int_0^1
 \frac{r(E+s(F-E))}{k-E-s(F-E)}\,ds\nonumber\\
 &=\frac1\pi\int_E^F\frac{r(\zeta)}{k-\zeta}\,d\zeta
 =\int_E^F\frac{2ir(\zeta)}{k-\zeta}\frac{d\zeta}{2\pi i}.
\end{align}
The uniformity follows because the denominator is bounded away from zero on compact subsets of $\mathbb C\setminus K_+$.  The second limit follows by complex conjugation and the parametrization of the oppositely oriented segment $\Sigma_2$.
\end{proof}

With the help of Proposition \ref{p1}, the jump conditions of $Z(k)$ transform into the following formula as $N\to\infty$
\begin{align}\label{Zjump-limit}
Z_{+} (k)= Z_{-}(k)
\begin{cases}
\begin{pmatrix}1&0\\ \displaystyle {e^{2 it\theta(k)}\int_{E}^{F} \frac{ 2ir(\zeta)}{k-\zeta }\frac{\d \zeta}{2\pi i}} & {1} \end{pmatrix},  \quad k \in  \Gamma_+,\\
\begin{pmatrix} {1}& -\displaystyle {e^{-2i t\theta(k)} \int_{F^*}^{E^*}\frac{2i\bar{r}(\zeta)}{k-\zeta }\frac{\d \zeta}{2\pi i} }\\ 0 & {1} \end{pmatrix}, \quad k \in \Gamma_-,
\end{cases}
\end{align}
where $\bar{r}(k)=r^*(k^*)$. Next, we define the following transformation to eliminate its jump on the contours $\Gamma_+$ and $\Gamma_-$
\begin{equation}
X(k) =Z(k)
\begin{cases}
 \begin{pmatrix} {1}&0 \\ \displaystyle{e^{ 2 it\theta(k)} \int_{E}^{F} \frac{2ir(\zeta)}{\zeta-k} \frac{\d \zeta}{2\pi i} } & {1}\end{pmatrix}
        ,&\text{$k$ inside the loop  $\Gamma_+$},\\
 \begin{pmatrix} {1}& \displaystyle{ e^{-2it\theta(k)} \int_{F^*}^{E^*} \frac{2i\bar{r}(\zeta)}{\zeta-k } \frac{\d \zeta}{2\pi i}}\\  0& {1} \end{pmatrix}
        ,&\text{$k$ inside the loop  $\Gamma_-$},\\
 I,
       &\text{elsewhere}.
\end{cases}
\end{equation}
 There are still jumps across $\Sigma_{1}\cup\Sigma_{2}$ because the integrals have jumps across those segments. By the Sokhotski-Plemelj formula, we obtain a RH problem for $X(k;x,t)$.
\begin{RHP}\label{RHP2} Find a matrix-valued function $X(k;x,t)$ with the following properties
\begin{enumerate}
\item  $X(k;x,t)$ is analytic for $k\in \C\backslash\le(\Sigma_{1}\cup\Sigma_{2} \ri)$.
 \item For $k \in \Sigma_{1}\cup\Sigma_{2}$, the boundary value $X_{+}(k;x,t)$ is taken from the left side of the contours and  the boundary value $X_{-}(k;x,t)$ is taken from the right and they satisfy the following jump relation
\begin{align}\label{gas1}
&X_{+}(k) = X_{-}(k) \begin{cases} \begin{pmatrix} {1}&0\\ \displaystyle {2ir(k) e^{ 2 it\theta(k)} } & {1} \end{pmatrix}, & k \in \Sigma_{1}, \\
\begin{pmatrix} {1}& \displaystyle {2i\bar{r}(k) e^{-2 it\theta(k)}}\\ 0 & {1} \end{pmatrix}, &  k \in \Sigma_{2},
\end{cases}
\end{align}
where the jump contours see Fig.\ref{contour1}.
\item
$
X(k;x,t) = I + \mathcal{O}\le(\frac{1}{k}\ri), \quad \text{as}\  k \to \infty.
$
 \end{enumerate}
\end{RHP}
We can recover the soliton gas $q(x,t)$ from
\begin{align}\label{u-X}
	q(x,t)=2i\lim\limits_{k\to\infty}kX_{12}(k;x,t).
\end{align}
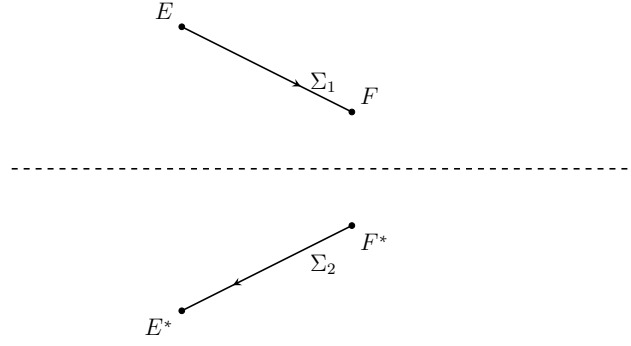
\begin{figure}
\centering
\scalebox{.75}{
\begin{tikzpicture}[>=stealth]op
\path (0,0) coordinate (O);
\node [above] at (-2.5,-0.3) { ${\Sigma}_1$};
\node [above] at (-2.5,-3.5){${\Sigma}_2$};
\draw[fill] (-2,-2.5) circle [radius=0.05];
\node[below right] at (-2,-2.5) {$F^*$};
\draw[fill] (-5,-4) circle [radius=0.05];
\node[below left] at (-5,-4) {$E^*$};
\draw[fill] (-5,1) circle [radius=0.05];
\node[above left] at (-5,1) {$E$};
\draw[fill] (-2,-0.5) circle [radius=0.05];
\node[above right] at (-2,-0.5) {$F$};
\draw[thick,dashed] (-8,-1.5)--(3,-1.5);
\draw[->- = .7,thick] (-5,1)--(-2,-0.5) ;
\draw[->- = .7,thick] (-2,-2.5)--(-5,-4) ;
\end{tikzpicture}
}
\caption{The oriented contours $\Sigma_{1}\cup\Sigma_{2}$.}
\label{contour1}
\end{figure}
We next establish the unique solvability of RHP~\ref{RHP2}.
For this purpose, reverse the orientation of $\Sigma_2$. The
corresponding jump matrix on $\Sigma_2$ is then replaced by its
inverse. Let
\[
V_1(k)=
\begin{pmatrix}
	1&0\\
	2ir(k)e^{2it\theta(k)}&1
\end{pmatrix},
\qquad k\in\Sigma_1.
\]
Since $x,t\in\mathbb R$, we have
\[
\theta(k^*)=\theta(k)^*,
\qquad
\bar r(k^*)=r(k)^*.
\]
Therefore, for $k\in\Sigma_1$,
\[
V_2(k^*)^{-1}
=
\begin{pmatrix}
	1&-2ir(k)^*e^{-2it\theta(k)^*}\\
	0&1
\end{pmatrix}
=
V_1(k)^\dagger,
\]
where $V_2$ denotes the original jump matrix on $\Sigma_2$.
Thus, after reversing the orientation of $\Sigma_2$, the jump
contour is invariant under Schwarz reflection, including its
orientation, and the jump matrix satisfies
\[
V(k^*)=V(k)^\dagger.
\]
Moreover, $\det V(k)=1$. Hence Zhou's vanishing lemma
\cite{Zhou1989} implies that the corresponding homogeneous RH
problem admits only the trivial solution. Since the associated
singular-integral operator is Fredholm with index zero, it follows
that the operator is invertible on
$L^2(\Sigma_1\cup\Sigma_2)$. Consequently, RHP~\ref{RHP2} has a
unique solution for every $(x,t)\in\mathbb R\times[0,+\infty)$.
We now return to the continuum limit. Let
$K\subset\mathbb R_x\times[0,+\infty)_t$ be compact. Since the jump
matrix of RHP~\ref{RHP2} depends continuously on $(x,t)$ and the
associated singular-integral operator is invertible for every
$(x,t)\in K$, the continuity of the inverse map together with the
compactness of $K$ gives
\[
\sup_{(x,t)\in K}
\left\|
(I-C_w)^{-1}
\right\|_{L^2\to L^2}<\infty.
\]
The Riemann-sum approximation in Proposition~\ref{p1} implies that
the singular-integral operators associated with the pole-removed
$N$-soliton RH problems converge to the limiting operator in
operator norm at the rate $\mathcal O(N^{-1})$, uniformly for
$(x,t)\in K$. Hence, for all sufficiently large $N$, these operators
are invertible uniformly on $K$, and the standard resolvent identity
gives
\begin{equation}\label{RHP-stability-N}
	Z_N(k;x,t)-Z(k;x,t)=\mathcal O(N^{-1}),
\end{equation}
uniformly for $(x,t)\in K$ and for $k$ in compact subsets separated
from the jump contours. The same estimate holds for the coefficient
of $k^{-1}$ at infinity. Therefore,
\begin{equation}\label{qN-to-q}
	q_N(x,t)\longrightarrow q(x,t),
	\qquad N\to\infty,
\end{equation}
locally uniformly in $(x,t)$.
  Since the jump matrices depend smoothly (in fact, analytically) on $x$ and $t$, the same singular-integral argument applies after differentiating with respect to these variables.  Thus, wherever the limiting RH problem is uniquely solvable, the reconstructed gas potential is a classical solution $q\in C^\infty_{\rm loc}(\mathbb R\times(0,+\infty))$ of \eqref{NLS}.  We use this RH reconstruction as the definition of the soliton-gas solution throughout the paper.
\subsection{Large positive-$x$ behavior}

We first derive the large positive-$x$ behavior directly from the
limiting soliton-gas RH problem.  Setting $t=0$ in RHP~\ref{RHP2}, we
have
\begin{align}
	X_{+}(k;x)=X_{-}(k;x)
	\begin{cases}
		\begin{pmatrix}
			1&0\\
			2ir(k)e^{2ikx}&1
		\end{pmatrix},
		& k\in\Sigma_1,\\[3mm]
		\begin{pmatrix}
			1&2i\bar r(k)e^{-2ikx}\\
			0&1
		\end{pmatrix},
		& k\in\Sigma_2.
	\end{cases}
\end{align}
Since $\Sigma_1$ is the segment joining $E$ and $F$ in the upper
half-plane, for $k\in\Sigma_1$ one has
\[
\Im k\geq b_0,
\qquad
b_0:=\min\{E_2,F_2\}>0.
\]
Therefore,
\begin{equation}
	\left|e^{2ikx}\right|
	=e^{-2x\Im k}
	\leq e^{-2b_0x},
	\qquad k\in\Sigma_1.
\end{equation}
Similarly, since $\Sigma_2$ is the Schwarz-conjugate segment in the
lower half-plane, for $k\in\Sigma_2$ one has
$\Im k\leq-b_0$, and hence
\begin{equation}
	\left|e^{-2ikx}\right|
	=e^{2x\Im k}
	\leq e^{-2b_0x},
	\qquad k\in\Sigma_2.
\end{equation}
Since $r$ and $\bar r$ are bounded on the compact contours
$\Sigma_1$ and $\Sigma_2$, it follows that the jump matrix $V_X$
satisfies
\begin{equation}\label{large-x-positive-jump}
	\|V_X-I\|_{L^2(\Sigma_1\cup\Sigma_2)
		\cap L^\infty(\Sigma_1\cup\Sigma_2)}
	=\mathcal{O}\!\left(e^{-2b_0x}\right),
	\qquad x\to+\infty.
\end{equation}
Hence RHP~\ref{RHP2} is a small-norm RH problem for sufficiently
large positive $x$.  By the standard small-norm RH theory,
\begin{equation}\label{large-x-positive-X}
	X(k;x)
	=
	I+\mathcal{O}\!\left(
	\frac{e^{-2b_0x}}{1+|k|}
	\right),
	\qquad x\to+\infty,
\end{equation}
uniformly for $k$ away from $\Sigma_1\cup\Sigma_2$.  In particular,
as $k\to\infty$,
\[
X_{12}(k;x)
=
\mathcal{O}\!\left(
\frac{e^{-2b_0x}}{k}
\right).
\]
Using the reconstruction formula \eqref{u-X}, we obtain
\begin{equation}\label{large-x-positive-q}
	q(x,0)
	=
	\mathcal{O}\!\left(e^{-2b_0x}\right),
	\qquad x\to+\infty.
\end{equation}
Thus the soliton-gas potential decays exponentially to zero as
$x\to+\infty$.
Having established the exponentially decaying behavior as $x\to+\infty$,
we next study the large negative-$x$ and large-time asymptotics of the
focusing NLS soliton gas.
\section{Asymptotics of the soliton gas $q(x,0)$ as $x\to -\infty$}
\label{sec:3}
In this section, we consider the asymptotic behavior of the focusing NLS soliton gas $q(x,0)$ for the large negative $x$. The asymptotic behavior of $q(x,t_0)$ for fixed $t_0$ can be similarly analyzed.

When $t=0$, we obtain the jump conditions for $X(k;x)$
\begin{align}
&X_{+}(k;x) = X_{-}(k;x) \begin{cases} \begin{pmatrix} {1}&0\\ \displaystyle {2ir(k) e^{2ikx} } & {1} \end{pmatrix}, & k \in \Sigma_1, \\
\begin{pmatrix} {1}& \displaystyle {2i\bar{r}(k) e^{ -2 i k x} }\\ 0 & {1} \end{pmatrix}, &  k \in \Sigma_2.
\end{cases}
\end{align}
Introducing a $g$-function which has the form
\begin{gather}\label{gk}
g(k)=\int_{E} ^{k}\frac{\zeta^2-e_1\zeta+c_1}{R(\zeta)}d\zeta\ ,
\end{gather}
where $e_1=\frac{E+F+E^*+F^*}{2}$, $R(k)=\sqrt{(k-E)(k-E^*)(k-F)(k-F^*)}$  and the constant $c_1$ is chosen so that
\begin{align}\label{intF}
\int_{F} ^{F^*}\frac{\zeta^2-e_1\zeta+c_1}{R(\zeta)}d\zeta=0,
\end{align}
this is
\begin{align}\label{c1}
c_1=-\dfrac{\int_{F} ^{F^*}\frac{\zeta^2-e_1\zeta}{R(\zeta)}d\zeta}{\int_{F} ^{F^*}\frac{1}{R(\zeta)}d\zeta}\in \R.
\end{align}
By construction, the function $g(k)$ has the asymptotics as $k\to \infty$,
\begin{align}
g(k)=k+g_{\infty}+\mathcal{O}\le(\frac{1}{k}\ri), \ \  k \to  \infty \, \label{gconstraint3},
\end{align}
where
\begin{gather}\label{ginfty}
g_{\infty}=(\int\limits_{E} ^{\infty}+\int\limits_{E^*} ^{\infty})(\frac{\zeta^2-e_1\zeta+c_1}{2R(\zeta)}-\frac{1}{2})d\zeta-E_1\in \R.
\end{gather}

To justify the signature distribution of $\operatorname{Im}g$, we impose
the following nonintersecting-spectrum condition.  Define
\begin{equation}\label{eq:elliptic-parameters-EF}
	m
	:=
	1-\frac{|E-F|^2}{|E-F^*|^2},
	\qquad
	s
	:=
	\frac{\bigl(\operatorname{Im}(E-F^*)\bigr)^2}{|E-F^*|^2}.
\end{equation}
Since $E,F\in\mathbb C_+$, one has $0<m<1$ and $0<s\leq 1$.

\begin{assumption}
	\label{assumption:nonintersecting-spectrum}
	The branch points $E,F,E^*,F^*$ satisfy
	\begin{equation}\label{eq:nonintersecting-condition-EF}
		\mathcal{D}(E,F)
		:=
		1-m+s-2\frac{\mathcal{E}(m)}{\mathcal{K}(m)}
		>0,
	\end{equation}
	where $\mathcal{K}(m)$ and $\mathcal{E}(m)$ denote the complete elliptic
	integrals of the first and second kinds, respectively.
\end{assumption}

Assumption~\ref{assumption:nonintersecting-spectrum} implies that the
two zeros of the normalized differential
\begin{equation}\label{eq:quasimomentum-differential}
	dg(k)
	=
	\frac{k^2-e_1k+c_1}{R(k)}\,dk,
	\qquad
	R(k)^2=(k-E)(k-E^*)(k-F)(k-F^*),
\end{equation}
form a nonreal conjugate pair.  Equivalently, the zero level set of
$\operatorname{Im}g$ consists of two Schwarz-conjugate spectral arcs:
an arc in $\mathbb C_+$ connecting $E$ to $F$, and its conjugate in
$\mathbb C_-$ connecting $E^*$ to $F^*$.  In particular, neither arc
intersects the real axis.

Strictly speaking, these spectral arcs  of $\Im g=0$ need not coincide with the
straight segments introduced in Assumption~\ref{assumption1}.
Nevertheless, since $r(k)$ is analytic and nonvanishing in a
neighborhood containing the original segments and the regions of
deformation, the jump contours may be deformed to the corresponding
spectral arcs without changing the associated Riemann--Hilbert
problem.  In what follows, we retain the notation $\Sigma_1$ and
$\Sigma_2$ for the deformed spectral arcs.

We now determine the sign of $\operatorname{Im}g$.  First, the choice
of $c_1$ in \eqref{c1} is precisely the vanishing-period
normalization for the differential \eqref{eq:quasimomentum-differential}.
Moreover,
\[
dg(k)=\left(1+\mathcal{O}(k^{-2})\right)dk,
\qquad k\to\infty,
\]
and hence
\begin{equation}\label{eq:g-large-k-signature}
	g(k)=k+g_\infty+\mathcal{O}(k^{-1}),
	\qquad g_\infty\in\mathbb R.
\end{equation}
Thus $dg$ is the normalized quasi-momentum differential on the
elliptic curve
\[
R(k)^2=(k-E)(k-E^*)(k-F)(k-F^*).
\]

By Lemma~2.2 of \cite{GravaJenkinsZhangZhang2026}, the corresponding quasi-momentum
satisfies
\begin{equation}\label{eq:g-upper-lower-sign}
	\operatorname{Im}g(k)>0,
	\qquad k\in\mathbb C_+\setminus\Sigma_1,
\end{equation}
and
\begin{equation}\label{eq:g-lower-sign}
	\operatorname{Im}g(k)<0,
	\qquad k\in\mathbb C_-\setminus\Sigma_2.
\end{equation}
Furthermore,
\begin{equation}\label{eq:g-zero-level-set}
	\operatorname{Im}g_\pm(k)=0,
	\qquad k\in\Sigma_1\cup\Sigma_2,
\end{equation}
while Schwarz symmetry gives
\[
g(k^*)=g(k)^*,
\qquad
\operatorname{Im}g(k)=0,
\quad k\in\mathbb R.
\]
Consequently, the signature table of $\operatorname{Im}g$ is the one
shown in Figure~\ref{signtable1}.

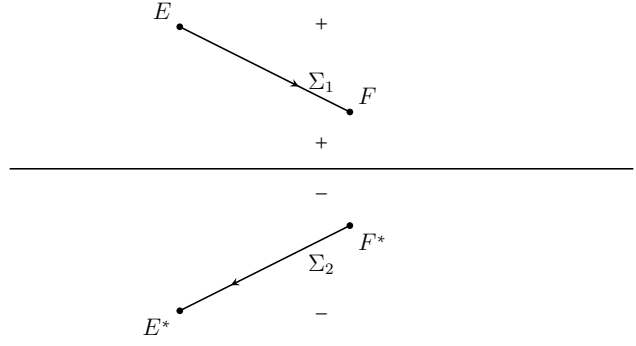
\begin{figure}
\centering
\scalebox{.75}{
\begin{tikzpicture}[>=stealth]op
\path (0,0) coordinate (O);
\node [above] at (-2.5,-0.3) { ${\Sigma}_1$};
\node [above] at (-2.5,-3.5){${\Sigma}_2$};
\draw[fill] (-2,-2.5) circle [radius=0.05];
\node[below right] at (-2,-2.5) {$F^*$};
\draw[fill] (-5,-4) circle [radius=0.05];
\node[below left] at (-5,-4) {$E^*$};
\draw[fill] (-5,1) circle [radius=0.05];
\node[above left] at (-5,1) {$E$};
\draw[fill] (-2,-0.5) circle [radius=0.05];
\node[above right] at (-2,-0.5) {$F$};
\draw[thick] (-8,-1.5)--(3,-1.5);
\draw[->- = .7,thick] (-5,1)--(-2,-0.5) ;
\draw[->- = .7,thick] (-2,-2.5)--(-5,-4) ;
\node[above] at (-2.5,0.8) {$+$};
\node[below] at (-2.5,-0.8) {$+$};
\node[below] at (-2.5,-3.8) {$-$};
\node[above] at (-2.5,-2.2) {$-$};
\end{tikzpicture}
}
\caption{The signature table of $\Im g(k)$.}
\label{signtable1}
\end{figure}
Then the function $g(k)$ is analytic in $\C\backslash ({\Sigma}_1\cup{\Sigma}_2\cup\Sigma_{F})$, with $\Sigma_{F}$ is a curve connecting the point $F$ and $F^*$, see Fig.~\ref{openinglenses1}, and satisfies
\begin{align}
& g_+(k) + g_-(k) =0 ,& k \in {\Sigma}_1\cup{\Sigma}_2 \label{gconstraint1}, \\
& g_+ (k)- g_-(k)=\Omega, & k \in \Sigma_{F}  \label{gconstraint2},
\end{align}
where
\begin{gather}\label{Omega}
\Omega=(\int\limits_{E}^{F}+\int\limits_{E^*}^{F^*})\frac{\zeta^2-e_1\zeta+c_1}{R_{+}(\zeta)}d\zeta\in \R.
\end{gather}
We also need the  following ansatz on the function $f(k)$:
\begin{align}
& f_+(k) f_-(k) =\frac{1}{2r(k)}, & k \in {\Sigma}_1 \label{fconstraint1}, \\
&f_+(k) f_-(k) = {2\bar{r}(k)}, & k \in {\Sigma}_2,\label{fconstraint2}\\
&\frac{f_+(k)}{ f_-(k) }=e^{i\Delta}, & k \in \Sigma_{F} \label{fconstraint3},\\
&f(k) = f_{\infty}+\mathcal{O}\le(\frac{1}{k}\ri), & k \to  \infty \label{fconstraint3}.
\end{align}
Then, the function $f(k)$ is given by
\begin{equation}
\label{f}
f(k)=\exp\le\{\frac{R(k)}{2\pi i}\left[\int_{ {\Sigma}_1}  \frac{\log(\frac{1}{2r(\zeta)}) }{R_+(\zeta)(\zeta-k)} \d \zeta+
\int_{{\Sigma}_2}  \frac{\log (2\bar{r}(\zeta)) }{R_+(\zeta)(\zeta-k)} \d \zeta+\int^{F^*}_{F}  \frac{i\Delta}{R(\zeta)(\zeta-k)} \d \zeta \ri]
\right\} \, .
\end{equation}
By \eqref{fconstraint3}, we obtain
\begin{align}
&\Delta=i\left[\int_{{\Sigma}_1} \frac{\log(\frac{1}{2r(\zeta)})}{R_+(\zeta)} \d \zeta+ \int_{{\Sigma}_2} \frac{\log (2\bar{r}(\zeta))}{R_+(\zeta)} \d \zeta \right]\left[\int^{F^*}_{F}\frac{\d\zeta}{R(\zeta)}\right]^{-1}\label{Delta},\\
&f_{\infty}=\exp\le\{-\frac{1}{2\pi i}\left[\int_{{\Sigma}_1}\frac{\log(\frac{1}{2r(\zeta)})(\zeta-e_1) }{R_+(\zeta)} \d \zeta+
\int_{{\Sigma}_2}  \frac{\log ({2\bar{r}(\zeta)})(\zeta-e_1)}{R_+(\zeta)} \d \zeta+\int^{F^*}_{F}  \frac{i\Delta(\zeta-e_1)}{R(\zeta)} \d \zeta \ri]
\right\}.\label{finfty}
\end{align}
We introduce a new matrix-valued function
\begin{equation}\label{EqD}
\begin{array}{l}
T(k)=f_{\infty}^{-\sigma_3}e^{-ixg_{\infty}\sigma_3}X(k)e^{ix(g(k)-k)\sigma_3}f(k)^{\sigma_3}.
\end{array}
\end{equation}
Then $T(k;x)$ satisfies the following RH problem.
\begin{RHP} Find a matrix-valued function $T(k;x)$ with the following properties
\begin{enumerate}
\item  $T(k;x)$ is analytic for $k\in \C\backslash ({\Sigma}_1\cup{\Sigma}_2\cup\Sigma_{F})$.
 \item  For $k\in {\Sigma}_1\cup{\Sigma}_2\cup\Sigma_{F}$, the boundary value $T_{+}(k;x)$ is taken from the left side of the contours and  the boundary value $T_{-}(k;x)$ is taken from the right, then we have the following jump relation
\begin{align}
\label{RH_T}
T_+(k)=T_-(k)\begin{cases}
\displaystyle \begin{pmatrix} e^{i x(g_+(k) - g_-(k))}\frac{f_+(k)}{f_-(k)} & 0\\ i & e^{-ix(g_+(k) - g_-(k))}\frac{f_-(k)}{f_+(k)} \end{pmatrix}, &\quad k \in  {\Sigma}_1,\\
\displaystyle \begin{pmatrix} e^{i x(g_+(k) - g_-(k))}\frac{f_+(k)}{f_-(k)}&i  \\ 0 & e^{-ix(g_+(k) - g_-(k))}\frac{f_-(k)}{f_+(k)} \end{pmatrix}, & \quad k \in {\Sigma}_2, \\
\displaystyle \begin{pmatrix} e^{i(x\Omega+\Delta) } & 0 \\0& e^{-i(x\Omega+\Delta)} \end{pmatrix},&\quad k \in  \Sigma_{F}.\\
  \end{cases}
  \end{align}
\item
$
T(k) = I + \mathcal{O}\le(\frac{1}{k}\ri), \quad \text{as} \ k \to \infty.
$
\end{enumerate}
\end{RHP}

\subsection{Opening lenses and the outer matrix parametrix $S^{\infty}(k)$} \label{TtoS}
The jump matrix for $T(k)$  on ${\Sigma}_{1}$  can be decomposed into
\begin{align*}
& \begin{pmatrix} e^{i x(g_+(k) - g_-(k))}\frac{f_+(k)}{f_-(k)}&0 \\ i & e^{-i x(g_+(k) - g_-(k))}\frac{f_-(k)}{f_+(k)} \end{pmatrix}\\
&=
\begin{pmatrix}1&-i\frac{e^{-2ixg_-(k)}}{2{r}(k)f_-^2 (k)}\\0 &1 \end{pmatrix}
\begin{pmatrix}0&i\\i&0\end{pmatrix}
\begin{pmatrix}1&-i\frac{e^{-2ixg_+(k)}}{2{r}(k)f_+^2(k)}\\0&1 \end{pmatrix}.
\end{align*}
and on ${\Sigma}_2$ as

\begin{align*}
&\begin{pmatrix} e^{i x\le(g_+(k) - g_-(k)\ri)}\dfrac{f_+(k)}{f_-(k)} & i \\ 0  & e^{-i x\le(g_+(k) - g_-(k)\ri)}\dfrac{f_-(k)}{f_+(k)} \end{pmatrix}\\
&=
\begin{pmatrix}1&0\\-i\frac{e^{2ixg_-(k)}f_-^2 (k)}{2\bar{r}(k)}&1 \end{pmatrix}
\begin{pmatrix}0&i \\i &0\end{pmatrix}
\begin{pmatrix}1&0\\ -i\frac{e^{2ixg_+(k)}f_+^2(k)}{2\bar{r}(k)}&1 \end{pmatrix},
\end{align*}
Define the function $\hat{r}(k)$  is the analytic  continuation of the functions $r(k)$ and $\bar{r}(k)$  off the contours ${\Sigma}_1\cup {\Sigma}_2$, which satisfies the following condition:
\begin{equation}\label{EqF}
\begin{array}{l}
\hat{r}(k)= r(k),\ k\in{\Sigma}_1,\ \  \hat{r}(k)= \bar{r}(k),\ k\in{\Sigma}_2.
\end{array}
\end{equation}
Using the decomposition above, we define a new matrix-valued function  $S(k)$ to proceed with opening lenses,
\begin{gather}
S(k) = \begin{cases}
\displaystyle T(k)\begin{pmatrix}1&i\frac{e^{-2ixg(k)}}{2\hat{r}(k)f^2 (k)}\\0 &1 \end{pmatrix},
 &\quad \text{in the upper lens, above ${\Sigma}_1$}, \\
\displaystyle T(k)\begin{pmatrix}1&-i\frac{e^{-2ixg(k)}}{2\hat{r}(k)f^2 (k)}\\0 &1 \end{pmatrix} ,
 &\quad \text{in the lower lens, below ${\Sigma}_1$}, \\
\displaystyle T(k)\begin{pmatrix}1&0\\ i\frac{e^{2i xg(k)}f^2(k)}{2\hat{r}(k)}&1 \end{pmatrix},
&\quad \text{in the lower lens, below ${\Sigma}_2$ },\\
\displaystyle T(k)\begin{pmatrix}1&0\\-i\frac{e^{2i x g(k)}f^2 (k)}{2\hat{r}(k)}&1 \end{pmatrix},
&\quad \text{in the upper lens, above ${\Sigma}_2$ },\\
\displaystyle T(k), &\quad \text{outside the lenses},
\end{cases} \label{RHPS}
\end{gather}
where the lenses are shown in Fig. \ref{openinglenses1}. It follows that the matrix-valued function $S(k)$ satisfies the following RH problem.
\begin{RHP}\label{RHP4} Find a matrix-valued function $S(k;y)$ with the following properties
\begin{enumerate}
\item $S(k;y)$ is analytic for $k\in\C\backslash \le({\Sigma}_1\cup{\Sigma}_2\cup\Sigma_{F}\cup {\mathcal C}_1\cup{\mathcal C}_2 \ri)$, the contours ${\mathcal C}_1\cup{\mathcal C}_2$ see Fig.\ref{openinglenses1}.
\item For $k\in {\Sigma}_1\cup{\Sigma}_2\cup\Sigma_{F}\cup {\mathcal C}_1\cup{\mathcal C}_2$, the boundary values $S_{\pm}(k;y)$ satisfy the jump relation
\begin{equation}
\label{VS}
\begin{split}
&S_+(k)=S_-(k)V_S(k),\\
&V_S(k)=\begin{cases}
\begin{pmatrix}1&-i\frac{e^{-2ixg(k)}}{2\hat{r}(k)f^2 (k)}\\0 &1 \end{pmatrix},\ \ k\in \mathcal{C}_1,\\
\begin{pmatrix}1&0\\-i\frac{e^{2i x g(k)}f^2 (k)}{2\hat{r}(k)}&1 \end{pmatrix},\ \ k\in \mathcal{C}_2,\\
\begin{pmatrix}0&i\\i &0 \end{pmatrix},\ \ k\in {\Sigma}_1\cup{\Sigma}_2,\\
\begin{pmatrix} e^{i(x\Omega+\Delta) } & 0 \\0& e^{-i(x\Omega+\Delta)} \end{pmatrix},\ \  k \in  \Sigma_{F}.
\end{cases}
\end{split}
\end{equation}

\item
$
S(k)=I + \mathcal{O}\le(\frac{1}{k}\ri), \quad k \to \infty.
$
\end{enumerate}
\end{RHP}
\begin{figure}
\centering
\scalebox{.75}{
\begin{tikzpicture}[>=stealth]op
\path (0,0) coordinate (O);
\node [above] at (-2.2,0.3) { ${\mathcal C}_1$};
\node [above] at (-3.2,-1.3) { ${\mathcal C}_1$};
\node [below] at (-3.2,-1.7) { ${\mathcal C}_2$};
\node [below] at (-2.2,-3.8) { $ {\mathcal C}_2$};
\node [below] at (-3.5,0.2) { ${\Sigma}_1$};
\node [above] at (-3.5,-4){${\Sigma}_2$};
\draw[fill] (-2,-2.5) circle [radius=0.05];
\node[above right] at (-2,-2.5) {$F^*$};
\draw[fill] (-5,-4) circle [radius=0.05];
\node[above left] at (-5,-4) {$E^*$};
\draw[fill] (-5,1) circle [radius=0.05];
\node[above left] at (-5,1) {$E$};
\draw[fill] (-2,-0.5) circle [radius=0.05];
\node[above right] at (-2,-0.5) {$F$};
\node[above right] at (-2,-1.5) {$\Sigma_{F}$};
\draw[->-=.7,thick,dashed] (-8,-1.5)--(3,-1.5);
\draw[->- = .7,thick] (-5,1)--(-2,-0.5) ;
\draw[->- = .7,thick] (-2,-2.5)--(-5,-4) ;
\draw[->-=.7,thick] (-2,-0.5)--(-2,-2.5);
\draw[->-=.7,thick] (-5,1)--(-3,1.25);
\draw[->-=.7,thick] (-3,1.25)--(-2,-0.5);
\draw[->-=.7,thick] (-5,1)--(-4,-0.75);
\draw[->-=.7,thick] (-4,-0.75)--(-2,-0.5);

\draw[->-=.7,thick] (-2,-2.5)--(-4,-2.25);
\draw[->-=.7,thick] (-4,-2.25)--(-5,-4);
\draw[->-=.7,thick] (-2,-2.5)--(-3,-4.25);
\draw[->-=.7,thick] (-3,-4.25)--(-5,-4);
\end{tikzpicture}
}
\caption{  Opening lenses.}
\label{openinglenses1}
\end{figure}
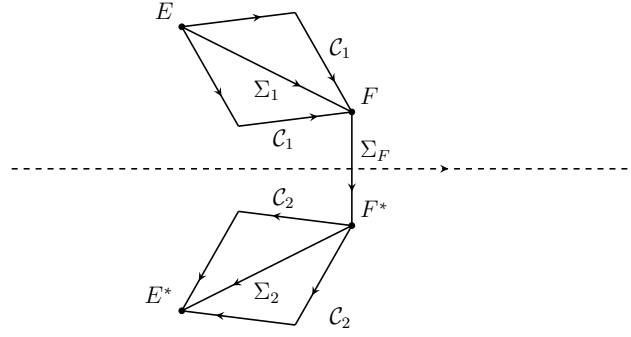
From the signature table of the function $\Im(g(k))$ as shown in Fig.\ref{signtable1}, we have
\begin{align}
\label{in1}
&\operatorname{Im} (g(k))>0\, ,\quad k\in {\mathcal C}_1\backslash\{E,F\},\\
\label{in2}
&\operatorname{Im} (g(k)) <0\, ,\quad  k\in {\mathcal C}_2\backslash\{F^*,E^*\}\, ,
\end{align}
where $\mathcal{C}_1$ and $\mathcal{C}_2$ are the contours defining the lenses as shown in Fig. \ref{openinglenses1}.
By above inequalities, the off-diagonal entries of the jump matrices for $S(k)$ exponentially vanish as $x\to -\infty$ along the upper and lower lenses. Then we get the model RH problem for $S^{\infty}(k)$.
\begin{RHP}\label{RHP5} Find a matrix-valued function $S^{\infty}(k;x)$ with the following properties
\begin{enumerate}
\item  $S^{\infty}(k;x)$ is analytic for $k\in \C\backslash ({\Sigma}_1\cup{\Sigma}_2\cup\Sigma_{F})$.
\item  For $k \in {\Sigma}_1\cup{\Sigma}_2\cup\Sigma_{F}$, the boundary values $S^{\infty}_{\pm}(k;x)$ satisfy the following jump relation
\begin{gather}
\label{Sinfinity1}
S^{\infty}_+(k) = S^{\infty}_-(k)
\begin{cases}
\begin{pmatrix} e^{i(x\Omega+\Delta)} &0 \\ 0 & e^{-i(x\Omega+\Delta)}\end{pmatrix},  & k \in \Sigma_{F},  \\
  \begin{pmatrix}0 & i\\ i & 0 \end{pmatrix}, &k \in {\Sigma}_1\cup {\Sigma}_{2}.  \\
\end{cases}
\end{gather}
\item
\begin{align}\label{Sinfinity2}
S^{\infty}(k)=I+\mathcal{O}\le(\frac{1}{k}\ri),\quad k\to\infty.
\end{align}
\end{enumerate}
\end{RHP}
To solve the RH problem \ref{RHP5}, we introduce a two-sheeted Riemann surface $\mathfrak{X}$ of genus $1$ associated to  the multivalued function $R(k)$,  namely
\[
\mathfrak{X}=\le\{ (k,\eta)\in\mathbb{C}^2\;|\; \eta^2=R^{2}(k)=(k-E)(k-E^*)(k-F)(k-F^*)\ri\} \, .
\]
The first sheet of $\mathfrak{X}$ is identified with the sheet where $R(k)=k^2(1+\mathcal{O}(\frac{1}{k}))$ as $k\to\infty^{+}$. The first and second sheet of the surface are two complex planes merged along the cuts ${\Sigma}_{1}$ and ${\Sigma}_{2}$. We introduce a canonical homology basis with the $\mathscr{B}$ cycle encircling ${\Sigma}_{1}$ clockwise on the first sheet and $\mathscr{A}$ cycle starts on the first sheet from the left side of the cut ${\Sigma}_{1}$, goes to the left side of cut ${\Sigma}_{2}$, proceeds to the lower sheet,  and coming back to ${\Sigma}_{1}$ on the second sheet. We introduce the holomorphic differential
 \begin{eqnarray}
\hat{\omega}=\left(2\int_{F}^{F^*}\frac{\mathrm{d}k}{R(k)}\right)^{-1}\frac{\mathrm{d}k}{R(k)},
\end{eqnarray}
so that
 \[
\oint_{\mathscr{A}}\hat{\omega}=1,
\]
and define the period
\[
\tau=\oint_{\mathscr{B}}\hat{\omega}=\left(\int_{F}^{F^*}\frac{\mathrm{d}k}{R(k)}\right)^{-1}\int_{E}^{F}\frac{\mathrm{d}k}{R(k)}.
\]
Then we introduce the Jacobi theta function
\begin{gather}
\vartheta_3 (z;\tau)= \sum_{n\in \mathbb{Z}} e^{2\pi i \, nz +  \pi  n^2 i \tau} \ , \quad z \in \mathbb{C} \ ,
\end{gather}
which is an even function of $z$ and satisfies the periodicity conditions
\begin{equation}
\label{periods}
\vartheta_3 (z+h+\hat{\lambda}\tau;\tau)=e^{-\pi  i \hat{\lambda}^2 \tau-2\pi i \hat{\lambda} z}\vartheta_3 (z;\tau)\, ,\quad h,\hat{\lambda}\in\mathbb{Z}.
\end{equation}
Using $\hat{\omega}$, we define the integral
\begin{gather}\label{Ak}
A(k)=\int_{E}^{k}\hat{\omega},\quad k\in\mathbb{C}\backslash ({\Sigma}_1\cup{\Sigma}_2\cup \Sigma_{F}).
\end{gather}
We observe that
\begin{equation}\label{Ajump}
\begin{array}{l}
A_{+}(k)+A_{-}(k)=0,\quad k\in{\Sigma}_{1},\\
A_{+}(k)-A_{-}(k)=\tau,\quad k\in \Sigma_{F},\\
A_{+}(k)+A_{-}(k)=1,\quad k\in{\Sigma}_{2}.
\end{array}
\end{equation}

Next we define the function
\begin{align}\label{gamma}
\gamma(k)=\le(\dfrac{(k-E^*)(k-F^*)}{(k-E)(k-F)}\ri)^{\frac14 },
\end{align}
with the properties
\begin{align}\label{gammajump}
\gamma_{+}(k)=i\gamma_{-}(k),\ k\in {\Sigma}_{1}\cup{\Sigma}_{2}\ \ \gamma_{+}(k)=-\gamma_{-}(k),\ k\in \Sigma_{F},\ \gamma(k)=1+\frac{i(E_2+F_2)}{2k}+\mathcal{O}(\frac{1}{k^2}),
\end{align}
where $E_2=\Im E,\ F_2=\Im F$.
The zeros of functions $\gamma(k)\pm\frac{1}{\gamma(k)}$ satisfy the equation $\gamma^4(k)=1$, which gives the zero
\begin{align}\label{E0}
E_0=\dfrac{EF-E^*F^*}{E+F-E^*-F^*}.
\end{align}
$E_0$ is real such that $\gamma^2(E_0)=-1$ and $\gamma(E_0)+\frac{1}{\gamma(E_0)}=0$.
Finally, we can construct the  solution of the RH problem \ref{RHP5}.
\begin{thm}
The matrix-valued solution $S^{\infty}(k)$ of RH problem \ref{RHP5} is considered on the first sheet of the Riemann surface $\mathfrak{X}$ with cut along the contours ${\Sigma}_1\cup{\Sigma}_2\cup [F,F^*]$ and the entries of $S^{\infty}(k)$ are given by
\begin{equation}
\label{Sinfty_sol}
\begin{split}
S^{\infty}_{11}(k)=\frac{1}{2}\le(\gamma(k)+\frac{1}{\gamma(k)}\ri)\dfrac{\vartheta_3\le(A(k)-A(E_0)-\frac{\tau}{2}-\frac{x\Omega+\Delta}{2\pi};\tau\ri)}{\vartheta_3\le( A(k)-A(E_0)-\frac{1}{2}-\frac{\tau}{2};\tau\ri)} \dfrac{\vartheta_3\le( A(\infty)-A(E_0)-\frac{1}{2}-\frac{\tau}{2};\tau\ri)}{\vartheta_3\le(A(\infty)-A(E_0)-\frac{\tau}{2}-\frac{x\Omega+\Delta}{2\pi};\tau\ri)},\\
S^{\infty}_{12}(k)=\frac{1}{2}\le(\gamma(k)-\frac{1}{\gamma(k)}\ri)\dfrac{\vartheta_3\le(A(k)+A(E_0)+\frac{\tau}{2}+\frac{x\Omega+\Delta}{2\pi};\tau\ri)}{\vartheta_3\le( A(k)+A(E_0)+\frac{1}{2}+\frac{\tau}{2};\tau\ri)} \dfrac{\vartheta_3\le( A(\infty)-A(E_0)-\frac{1}{2}-\frac{\tau}{2};\tau\ri)}{\vartheta_3\le(A(\infty)-A(E_0)-\frac{\tau}{2}-\frac{x\Omega+\Delta}{2\pi};\tau\ri)},\\
S^{\infty}_{21}(k)=\frac{1}{2}\le(\gamma(k)-\frac{1}{\gamma(k)}\ri)\dfrac{\vartheta_3\le(A(k)+A(E_0)+\frac{\tau}{2}-\frac{x\Omega+\Delta}{2\pi};\tau\ri)}{\vartheta_3\le( A(k)+A(E_0)+\frac{1}{2}+\frac{\tau}{2};\tau\ri)} \dfrac{\vartheta_3\le( A(\infty)-A(E_0)-\frac{1}{2}-\frac{\tau}{2};\tau\ri)}{\vartheta_3\le(A(\infty)-A(E_0)-\frac{\tau}{2}+\frac{x\Omega+\Delta}{2\pi};\tau\ri)},\\
S^{\infty}_{22}(k)=\frac{1}{2}\le(\gamma(k)+\frac{1}{\gamma(k)}\ri)\dfrac{\vartheta_3\le(A(k)-A(E_0)-\frac{\tau}{2}+\frac{x\Omega+\Delta}{2\pi};\tau\ri)}{\vartheta_3\le( A(k)-A(E_0)-\frac{1}{2}-\frac{\tau}{2};\tau\ri)}\dfrac{\vartheta_3\le( A(\infty)-A(E_0)-\frac{1}{2}-\frac{\tau}{2};\tau\ri)}{\vartheta_3\le(A(\infty)-A(E_0)-\frac{\tau}{2}+\frac{x\Omega+\Delta}{2\pi};\tau\ri)}.
\end{split}
\end{equation}
\end{thm}
\begin{proof}
Using the properties of function $\gamma(k)$ in \eqref{gammajump} and the Abel integral $A(k)$ in \eqref{Ajump}, we have that $S^{\infty}_{11+}(k)=iS^{\infty}_{12-}(k)$ for $k\in {\Sigma}_1\cup{\Sigma}_2$ and $S^{\infty}_{11+}(k)=e^{i(x\Omega+\Delta)}S^{\infty}_{11-}(k)$ for $k\in \Sigma_{F}$ and so on. As $k\to \infty$, $S^{\infty}(k)\to I$, namely the condition \eqref{Sinfinity2} is satisfied.
\end{proof}
Note that  $S^{\infty}_{12}(k)$ and $S^{\infty}_{21}(k)$ are analytic and bounded on the first sheet. The poles of $S^{\infty}_{11}(k)$ and $S^{\infty}_{22}(k)$ on the first sheet are eliminated by the factor $\gamma(k)+\frac{1}{\gamma(k)}$. Since $\det S^{\infty}(k)$ does not have jumps on ${\Sigma}_{1}\cup{\Sigma}_{2}\cup\Sigma_{F}$ and $S^{\infty}(k)\to I$ as $k\to\infty$, we have that $\det S^{\infty}(k)\equiv 1$.

The matrix-valued solution $S^{\infty}(k)$ provides the asymptotic behaviour of the RH problem $S(k)$ for all $k$  away from the endpoints $E,F,E^*,F^*$. The local representation of $g(k)$ at the endpoints $E,F,E^*$ and $F^*$ exhibits a square root type behavior:
\begin{align}
g(k)=\mathcal{O}((k-\alpha)^{\frac12}),\ k\to \alpha, \ \alpha=E,E^*,\ \ g(k)=\frac{\Omega}{2}+\mathcal{O}((k-\beta)^{\frac12}),\ k\to \beta, \ \beta=F,F^*,
\end{align}
this means that the associated local RH problems near the endpoints $E,F,E^*$ and $F^*$ can be solved in terms of the modified Bessel functions. Next we show the construction of the local parametrix $P^{E^*}(k)$ around $E^*$. We define the local conformal map:
\begin{align}\label{zetak}
\zeta(k)=\frac{1}{4}[ixg(k)]^2,
\end{align}
and we note that $\zeta\in \R_{-}$ when $k\in {\Sigma}_{2}$.
We fix a small disc $B_{\rho}^{E^*}=\{ k\in\C\big| |k-E^*|<\rho\}$ center at $E^*$ of sufficiently small radius $\rho$, and we choose the contours $\mathcal{C}_2$ and ${\Sigma}_{2}$ within $B_{\rho}^{E^*}$ so that their images in $\zeta(B_{\rho}^{E^*})$ lie respectively on the straight rays $\arg (\zeta)=\pm \frac{2}{3}\pi$ and $\arg (\zeta)=\pi$. We make the transformation by defining
\begin{align}
P^{(1)}(k(\zeta))=S(k(\zeta))\le(\frac{\sqrt{2\hat{r}(k(\zeta))}e^{\frac{i}{4}\pi}}{f(k(\zeta))}\ri)^{\sigma_3}e^{-2\zeta^{\frac12}\sigma_3}
,\ \zeta\in \C.
\end{align}
Consequently, the matrix-valued function $P^{(1)}(\zeta)$ satisfies the local simple jump conditions.
\begin{align}
P^{(1)}_{+}(\zeta)=P^{(1)}_{-}(\zeta)\begin{cases}
\begin{pmatrix}
1&0\\1&1
\end{pmatrix},\ \ \zeta\in (\arg\zeta=\pm \frac{2}{3}\pi) \cap \zeta(B_{\rho}^{E^*}),\\
\begin{pmatrix}
0&1\\-1&0
\end{pmatrix},\ \ \zeta\in (-\infty,0]\cap \zeta(B_{\rho}^{E^*}).
\end{cases}
\end{align}
Then we introduce the model parametrix $\Psi_{Bes}(\zeta)$ as in \cite{Girotti-1}, which satisfies the following jump relation
\begin{align}
\Psi_{Bes+}(\zeta)=\Psi_{Bes-}(\zeta)\begin{cases}
\begin{pmatrix}
1&0\\1&1
\end{pmatrix},\ \  \arg \zeta=\pm \frac{2\pi}{3} ,\\
\begin{pmatrix}
0&1\\-1&0
\end{pmatrix},\ \ \arg \zeta=\pi,
\end{cases}
\end{align}
with asymptotics at infinity
\begin{align}
\Psi_{Bes}(\zeta)=(2\pi \zeta^{\frac12})^{-\frac12 \sigma_3}\frac{1}{\sqrt{2}}\begin{pmatrix}1&i\\i&1 \end{pmatrix}\le(  I +\mathcal{O}(\frac{1}{\zeta^{\frac12}})\ri)e^{2\zeta^{\frac12}\sigma_3}.
\end{align}
Then we find the local parametrix around the point $E^*$ is
\begin{gather}
P^{E^*}(k)=D(k)\Psi_\Bes(\zeta(k)) e^{-2\zeta(k)^{\frac{1}{2}}\sigma_3} \left(
\frac{\sqrt{2 \hat{r}(k)} e^{ i \pi  /4}}{ f(k)}
\right)^{-\sigma_{3}},\ k \in B^{E^*}_{\rho},
\end{gather}
where
\begin{gather}
D(k)
=S^{\infty}(k)
\left(
\frac{ \sqrt{2\hat{r}(k)}e^{i\pi/4}}{f(k)}
\right)^{\sigma_{3}}
\frac{1}{\sqrt{2}} \begin{pmatrix}1& -i \\ -i  &  1\end{pmatrix} \le(2\pi \zeta^{\frac{1}{2}}\ri)^{\frac{1}{2}\sigma_3}, \quad  k \in B^{E^*}_{\rho}.
\end{gather}
Therefore, we have the following matching condition on the boundary $\partial B^{E^*}_{\rho} $:
\begin{gather}
P^{E^*}(k)(S^{\infty}(k))^{-1}= I +\mathcal{O}(\frac{1}{x}),\  \text{as}\ x\to -\infty.
\end{gather}
The local parametrix of other endpoints can be constructed similarly.
Define $P(k)$ to be equal $P^{\alpha}(k)$ inside the disk $B_{\rho}^{\alpha}$, $\alpha=E,E^*,F,F^*$, and define $P(k)$ to be equal to $S^{\infty}(k)$ elsewhere. Next, define the error function
\begin{gather}
	 {\mathcal{E}}(k) =  {S}(k) \left( {P}(k) \right)^{-1}.
\end{gather}
The inequalities \eqref{in1}--\eqref{in2} and the matching properties of $P^{\alpha}(k)$, $\alpha=E,E^*,F,F^*$ allow us to conclude that
\begin{align}
{\mathcal{E}}(k)=I+\mathcal{O}(\frac{1}{x}),\ \ \text{as}\ x\to-\infty.
\end{align}
Taking into account all the transformations we performed, we are now able to explicitly solve the original RH problem for $X(k;x)$ as $x\to -\infty$ and $k\to \infty$,
\begin{equation}
\begin{split}
X(k)&= f_{\infty}^{\sigma_3}e^{ixg_{\infty}\sigma_3}T(k) e^{-ix(g(k)-k) \sigma_3}f(k)^{-\sigma_3}\\
 &= f_{\infty}^{\sigma_3}e^{ixg_{\infty}\sigma_3}S(k) e^{-ix(g(k)-k) \sigma_3}f(k)^{-\sigma_3} \\
&=f_{\infty}^{\sigma_3}e^{ixg_{\infty}\sigma_3}(I+\mathcal{O}(\frac{1}{x}))S^{\infty}(k) e^{-ix(g(k)-k) \sigma_3}f(k)^{-\sigma_3}.
\end{split}
\end{equation}
From \eqref{u-X}, we obtain
\begin{equation}
q(x,0)=2if^2_{\infty}e^{2ixg_{\infty}}\le(\lim\limits_{k\to\infty}kS_{12}^{\infty}(k)\ri)+\mathcal{O}(\frac{1}{x}).
\end{equation}
Taking into account that
\begin{align*}
2i\lim\limits_{k\to\infty}kS_{12}^{\infty}(k)=-(E_2+F_2)\dfrac{\vartheta_3\le(A(\infty)+A(E_0)+\frac{\tau}{2}+\frac{x\Omega+\Delta}{2\pi};\tau\ri)}{\vartheta_3\le( A(\infty)+A(E_0)+\frac{1}{2}+\frac{\tau}{2};\tau\ri)} \dfrac{\vartheta_3\le( A(\infty)-A(E_0)-\frac{1}{2}-\frac{\tau}{2};\tau\ri)}{\vartheta_3\le(A(\infty)-A(E_0)-\frac{\tau}{2}-\frac{x\Omega+\Delta}{2\pi};\tau\ri)},
\end{align*}
we get the asymptotics of the soliton gas $q(x,0)$ for $x\to-\infty$, the asymptotic formula  see \eqref{qx}.

\section{Behaviour of the potential $q(x,t)$ as $t \to +\infty$}
\label{sec:4}
Recall the RH problem associated with the soliton gas for the focusing NLS equation:
\begin{align}
	&X_{+}(k;x,t)=X_{-}(k;x,t)
	\begin{cases}
		\begin{pmatrix}
			1&0\\[0.5ex]
			2i\,r(k)e^{2it\theta(k)}&1
		\end{pmatrix}, & k\in \Sigma_1,\\[3ex]
		\begin{pmatrix}
			1&2i\bar r(k)e^{-2it\theta(k)}\\[0.5ex]
			0&1
		\end{pmatrix}, & k\in \Sigma_2,
	\end{cases}\\
	&X(k;x,t)=I+\mathcal{O}\!\left(\frac{1}{k}\right),\qquad k\to\infty,
\end{align}
where the phase function is given by $\theta(k)=k^2+2k\xi$ and $\xi=\frac{x}{2t}\in\mathbb{R}$.

In the region $\xi>-E_1$, where $E_1=\Re E$, the off-diagonal entries of the jump matrices are exponentially small as $t\to+\infty$. Therefore, by a standard small-norm argument, one obtains
\begin{gather}
	X(k;x,t)=I+\mathcal{O}\!\left(e^{-ct}\right),\qquad t\to+\infty,\quad c>0,\quad \xi>-E_1,
\end{gather}
and hence the corresponding potential $q(x,t)$ is asymptotically trivial.

The more interesting case is $\xi< -E_1$. To analyze this regime, we first introduce the following system of equations:
\begin{subequations}\label{xt2}
	\begin{align}
		&\mu=E_1-H_1-\xi,\label{system21}\\
		&H_2^2=E_2^2-2(E_1-\mu)(E_1-H_1),\label{system22}\\
		&\int\limits_{E}^{E^*}(k-\mu)\sqrt{\frac{(k-H)(k-H^*)}{(k-E)(k-E^*)}}\,\d k=0,\label{system23}
	\end{align}
\end{subequations}
where $\mu,H_1,H_2\in\mathbb{R}$ and $H=H_1+iH_2$. For $\xi\in(-E_1-\sqrt2\,E_2,-E_1)$, the system \eqref{xt2} admits a unique solution $\mu(\xi)$, $H_1(\xi)$, and $H_2(\xi)$; this will be proved in the next section.

In what follows, we require that the point $F$ lie on the trajectory traced by
$
H(\xi)=H_1(\xi)+iH_2(\xi),
$
that is, there exists a $\hat\xi\in(-E_1-\sqrt2\,E_2,-E_1)$ such that
$
H(\hat\xi)=F.
$
A typical trajectory of the point $H(\xi)$ for $\xi\in(-E_1-\sqrt2\,E_2,-E_1)$ is illustrated in Fig.~\ref{trajectory}.
\begin{figure}[htbp]
	\centering
	\begin{tikzpicture}
		\begin{axis}[
			width=0.5\textwidth,
			height=0.5\textwidth,
			xlabel={$\Re H=H_1$},
			ylabel={$\Im H=H_2$},
			xmin=-4.15,
			xmax=-1.05,
			ymin=-0.10,
			ymax=4.20,
			xtick={-4.0,-3.5,-3.0,-2.5,-2.0,-1.5},
			ytick={0,0.5,1.0,1.5,2.0,2.5,3.0,3.5,4.0},
			axis lines=box,
			tick align=outside,
			tick style={black},
			]
			\addplot[
			blue,
			very thick,
			smooth,
			]
			coordinates {
				(-4.00434,4.00154)
				(-3.99566,4.00154)
				(-3.99276,3.99930)
				(-3.93489,3.99930)
				(-3.93200,3.99707)
				(-3.91753,3.99707)
				(-3.91464,3.99484)
				(-3.88860,3.99484)
				(-3.88570,3.99261)
				(-3.87702,3.99261)
				(-3.87413,3.99038)
				(-3.85388,3.99038)
				(-3.85098,3.98815)
				(-3.83073,3.98592)
				(-3.82783,3.98369)
				(-3.81337,3.98369)
				(-3.81047,3.98146)
				(-3.80469,3.98146)
				(-3.80179,3.97923)
				(-3.78732,3.97923)
				(-3.77864,3.97477)
				(-3.76418,3.97477)
				(-3.75550,3.97030)
				(-3.74682,3.97030)
				(-3.74392,3.96807)
				(-3.73814,3.96807)
				(-3.73524,3.96584)
				(-3.72945,3.96584)
				(-3.72656,3.96361)
				(-3.71788,3.96361)
				(-3.71499,3.96138)
				(-3.70920,3.96138)
				(-3.70631,3.95915)
				(-3.69763,3.95915)
				(-3.68895,3.95469)
				(-3.68026,3.95469)
				(-3.67737,3.95246)
				(-3.67158,3.95246)
				(-3.66290,3.94800)
				(-3.65422,3.94800)
				(-3.64844,3.94353)
				(-3.63976,3.94353)
				(-3.63686,3.94130)
				(-3.63397,3.94130)
				(-3.63108,3.93907)
				(-3.62529,3.93907)
				(-3.61661,3.93461)
				(-3.60214,3.93238)
				(-3.59346,3.92792)
				(-3.58767,3.92792)
				(-3.58478,3.92569)
				(-3.58189,3.92569)
				(-3.57899,3.92346)
				(-3.57320,3.92346)
				(-3.55874,3.91676)
				(-3.55295,3.91676)
				(-3.54716,3.91230)
				(-3.53848,3.91230)
				(-3.53559,3.90784)
				(-3.52691,3.90784)
				(-3.52112,3.90338)
				(-3.51533,3.90338)
				(-3.48929,3.89223)
				(-3.48351,3.89223)
				(-3.47193,3.88553)
				(-3.46614,3.88553)
				(-3.46036,3.88107)
				(-3.45457,3.88107)
				(-3.44878,3.87661)
				(-3.44300,3.87661)
				(-3.43721,3.87215)
				(-3.43142,3.87215)
				(-3.42564,3.86769)
				(-3.39959,3.85876)
				(-3.39381,3.85430)
				(-3.38802,3.85430)
				(-3.37355,3.84538)
				(-3.36777,3.84538)
				(-3.36487,3.84092)
				(-3.35908,3.84092)
				(-3.35330,3.83645)
				(-3.35040,3.83645)
				(-3.34751,3.83422)
				(-3.33883,3.83199)
				(-3.33304,3.82753)
				(-3.31858,3.82307)
				(-3.30411,3.81415)
				(-3.29832,3.81415)
				(-3.29543,3.80969)
				(-3.28964,3.80969)
				(-3.28675,3.80522)
				(-3.28096,3.80522)
				(-3.27517,3.80076)
				(-3.27228,3.80076)
				(-3.26649,3.79630)
				(-3.26360,3.79630)
				(-3.24045,3.78292)
				(-3.23466,3.78292)
				(-3.23177,3.77845)
				(-3.22598,3.77845)
				(-3.22309,3.77399)
				(-3.21730,3.77399)
				(-3.21441,3.76953)
				(-3.21152,3.76953)
				(-3.20573,3.76507)
				(-3.20283,3.76507)
				(-3.17969,3.75168)
				(-3.17679,3.75168)
				(-3.17390,3.74722)
				(-3.16811,3.74722)
				(-3.16522,3.74276)
				(-3.15943,3.74276)
				(-3.15654,3.73830)
				(-3.15365,3.73830)
				(-3.14786,3.73384)
				(-3.14496,3.73384)
				(-3.13918,3.72938)
				(-3.13628,3.72938)
				(-3.13050,3.72491)
				(-3.12760,3.72491)
				(-3.12182,3.72045)
				(-3.11892,3.72045)
				(-3.11603,3.71599)
				(-3.11024,3.71599)
				(-3.10735,3.71153)
				(-3.10446,3.71153)
				(-3.10156,3.70707)
				(-3.09577,3.70707)
				(-3.09288,3.70261)
				(-3.08999,3.70261)
				(-3.08420,3.69814)
				(-3.08131,3.69814)
				(-3.06973,3.68922)
				(-3.06684,3.68922)
				(-3.06105,3.68476)
				(-3.05816,3.68476)
				(-3.05527,3.68030)
				(-3.05237,3.68030)
				(-3.04948,3.67584)
				(-3.04369,3.67584)
				(-3.04080,3.67137)
				(-3.03790,3.67137)
				(-3.02633,3.66245)
				(-3.02344,3.66245)
				(-3.01765,3.65799)
				(-3.01476,3.65799)
				(-3.01186,3.65353)
				(-3.00897,3.65353)
				(-3.00608,3.64907)
				(-3.00029,3.64907)
				(-2.99740,3.64461)
				(-2.99450,3.64461)
				(-2.99161,3.64014)
				(-2.98871,3.64014)
				(-2.98293,3.63568)
				(-2.98003,3.63568)
				(-2.97714,3.63122)
				(-2.97425,3.63122)
				(-2.96267,3.62230)
				(-2.95978,3.62230)
				(-2.95689,3.61784)
				(-2.95399,3.61784)
				(-2.95110,3.61337)
				(-2.94531,3.61337)
				(-2.94242,3.60891)
				(-2.93953,3.60891)
				(-2.93663,3.60445)
				(-2.93374,3.60445)
				(-2.93084,3.59999)
				(-2.92795,3.59999)
				(-2.92216,3.59553)
				(-2.91927,3.59553)
				(-2.91638,3.59107)
				(-2.91348,3.59107)
				(-2.91059,3.58660)
				(-2.90770,3.58660)
				(-2.90480,3.58214)
				(-2.90191,3.58214)
				(-2.89612,3.57768)
				(-2.89323,3.57768)
				(-2.89034,3.57322)
				(-2.88744,3.57322)
				(-2.88455,3.56876)
				(-2.88165,3.56876)
				(-2.87876,3.56430)
				(-2.87587,3.56430)
				(-2.86429,3.55537)
				(-2.86140,3.55537)
				(-2.85851,3.55091)
				(-2.85561,3.55091)
				(-2.85272,3.54645)
				(-2.84983,3.54645)
				(-2.84693,3.54199)
				(-2.84404,3.54199)
				(-2.84115,3.53753)
				(-2.83825,3.53753)
				(-2.83536,3.53306)
				(-2.83247,3.53306)
				(-2.82089,3.52414)
				(-2.81800,3.52414)
				(-2.81510,3.51968)
				(-2.80353,3.51076)
				(-2.80064,3.51076)
				(-2.79774,3.50629)
				(-2.79485,3.50629)
				(-2.79196,3.50183)
				(-2.78906,3.50183)
				(-2.78617,3.49737)
				(-2.78328,3.49737)
				(-2.78038,3.49291)
				(-2.77749,3.49291)
				(-2.77459,3.48845)
				(-2.77170,3.48845)
				(-2.76881,3.48399)
				(-2.76591,3.48399)
				(-2.76302,3.47952)
				(-2.76013,3.47952)
				(-2.75723,3.47506)
				(-2.75434,3.47506)
				(-2.75145,3.47060)
				(-2.74855,3.47060)
				(-2.74566,3.46614)
				(-2.74277,3.46614)
				(-2.73987,3.46168)
				(-2.73698,3.46168)
				(-2.73409,3.45722)
				(-2.73119,3.45722)
				(-2.72830,3.45276)
				(-2.71094,3.43937)
				(-2.70804,3.43491)
				(-2.70515,3.43491)
				(-2.70226,3.43045)
				(-2.69936,3.43045)
				(-2.69647,3.42599)
				(-2.69358,3.42599)
				(-2.69068,3.42152)
				(-2.68779,3.42152)
				(-2.68490,3.41706)
				(-2.68200,3.41706)
				(-2.67911,3.41260)
				(-2.67622,3.41260)
				(-2.66464,3.39922)
				(-2.66175,3.39922)
				(-2.65885,3.39475)
				(-2.65596,3.39475)
				(-2.65307,3.39029)
				(-2.65017,3.39029)
				(-2.64728,3.38583)
				(-2.62992,3.37245)
				(-2.62703,3.36798)
				(-2.62413,3.36798)
				(-2.62124,3.36352)
				(-2.61835,3.36352)
				(-2.61545,3.35906)
				(-2.60388,3.35014)
				(-2.60098,3.34568)
				(-2.58941,3.33675)
				(-2.58652,3.33229)
				(-2.58362,3.33229)
				(-2.58073,3.32783)
				(-2.57784,3.32783)
				(-2.57494,3.32337)
				(-2.56337,3.31444)
				(-2.55179,3.30106)
				(-2.54890,3.30106)
				(-2.54601,3.29660)
				(-2.54311,3.29660)
				(-2.53154,3.28321)
				(-2.51997,3.27429)
				(-2.51707,3.26983)
				(-2.51418,3.26983)
				(-2.50839,3.26091)
				(-2.50550,3.26091)
				(-2.50260,3.25644)
				(-2.48235,3.23860)
				(-2.47946,3.23414)
				(-2.47656,3.23414)
				(-2.47078,3.22521)
				(-2.46788,3.22521)
				(-2.43895,3.19398)
				(-2.43605,3.19398)
				(-2.42737,3.18283)
				(-2.41869,3.17613)
				(-2.39844,3.15383)
				(-2.38976,3.14713)
				(-2.38108,3.13598)
				(-2.36661,3.12259)
				(-2.35793,3.11144)
				(-2.34925,3.10475)
				(-2.34346,3.09582)
				(-2.32899,3.08244)
				(-2.31453,3.06459)
				(-2.30874,3.06013)
				(-2.27980,3.02667)
				(-2.27402,3.02221)
				(-2.26823,3.01328)
				(-2.26244,3.00882)
				(-2.21615,2.95528)
				(-2.21036,2.94636)
				(-2.20457,2.94190)
				(-2.18432,2.91513)
				(-2.17853,2.91067)
				(-2.15538,2.87944)
				(-2.15249,2.87944)
				(-2.14092,2.86159)
				(-2.13513,2.85713)
				(-2.12355,2.83928)
				(-2.11777,2.83482)
				(-2.10619,2.81697)
				(-2.10330,2.81697)
				(-2.08883,2.79467)
				(-2.08305,2.79020)
				(-2.07147,2.77236)
				(-2.06279,2.76343)
				(-2.04543,2.73666)
				(-2.03964,2.73220)
				(-2.01649,2.69651)
				(-2.01071,2.69205)
				(-1.98467,2.65189)
				(-1.97309,2.63851)
				(-1.89207,2.51581)
				(-1.87761,2.49127)
				(-1.86892,2.48012)
				(-1.76765,2.31058)
				(-1.76186,2.29719)
				(-1.73582,2.25258)
				(-1.73293,2.24365)
				(-1.70689,2.19904)
				(-1.70399,2.19011)
				(-1.69531,2.17673)
				(-1.69242,2.16781)
				(-1.68374,2.15442)
				(-1.59404,1.96926)
				(-1.52749,1.81534)
				(-1.49277,1.72834)
				(-1.47251,1.67257)
				(-1.46962,1.66810)
				(-1.44068,1.58780)
				(-1.40018,1.46733)
				(-1.35677,1.32456)
				(-1.31626,1.17509)
				(-1.26997,0.97209)
				(-1.24682,0.85163)
				(-1.22656,0.72893)
				(-1.20631,0.57947)
				(-1.20052,0.52370)
				(-1.19184,0.43670)
				(-1.17738,0.19354)
				(-1.16580,0.07084)
			};
		\end{axis}
	\end{tikzpicture}
	\caption{Numerical simulation of the trajectory of $H(\xi)$ in the complex plane for $E_1=-4$ and $E_2=4$..}
	\label{trajectory}
\end{figure}
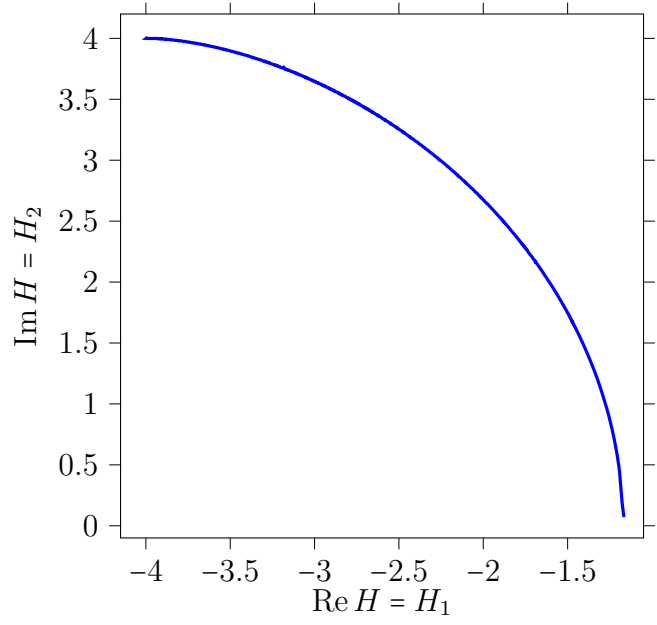
\section{The modulated elliptic wave region}
\label{sec:5}
In this section, we consider the region $\hat\xi<\xi<-E_1$. This restriction is imposed to ensure the existence of a suitable $g$-function.

We begin by introducing the scalar function $\tilde g(k)$:
\begin{align}
	\label{gprime1}
	&\tilde g(k)=\int\limits_E^k \d\tilde g(\zeta),\\
	&\d\tilde g(k)=2(k-\mu)\sqrt{\frac{(k-H)(k-H^*)}{(k-E)(k-E^*)}}\,\d k,
\end{align}
which is analytic in $\mathbb{C}\setminus(\tilde\Sigma_1\cup\tilde\Sigma_2)$, where the contours $\tilde\Sigma_1\cup\tilde\Sigma_2$ are shown in Fig.~\ref{signtable2}. The parameters  $H=H_1+H_2i,\ H_1=\Re H,\ H_2=\Im H$, $\mu$ are chosen to satisfy the following conditions:
\begin{enumerate}
\item $\int\limits_{H}^{H^*}\d \tilde{g}(k)=0$.
\item Asymptotics at infinity:
\begin{align}\label{gtinfy}
\tilde{g}(k)-\theta(k)=\tilde{g}_{\infty}+\mathcal{O}(k^{-1}),\quad k\to\infty.
\end{align}
\end{enumerate}
These conditions imply that the parameters $\mu$, $H_1$, and $H_2$ satisfy
\begin{subequations}\label{xt1}
	\begin{align}
		&\mu=E_1-\xi-H_1,\label{system11}\\
		&H_2^2-2\Bigl(H_1+\frac{\xi-E_1}{2}\Bigr)^2=E_2^2-\frac{(E_1+\xi)^2}{2},\label{system12}\\
		&\int\limits_H^{H^*}(k-\mu)\sqrt{\frac{(k-H)(k-H^*)}{(k-E)(k-E^*)}}\,\d k=0,\label{system13}
	\end{align}
	where $E=E_1+iE_2$ and $H=H_1+iH_2$.
\end{subequations}

Moreover, \eqref{gtinfy} shows that the residue of $\frac{\d\tilde g}{\d k}$ at infinity vanishes. Consequently, condition \eqref{system13} can be rewritten as
\begin{align}
	\int\limits_E^{E^*}(k-\mu)\sqrt{\frac{(k-H)(k-H^*)}{(k-E)(k-E^*)}}\,\d k=0.
\end{align}
Therefore, the system \eqref{xt1} is equivalent to \eqref{xt2}.

In order for the $g$-function to be appropriate for the asymptotic analysis, the parameters in its definition must satisfy the nonlinear system \eqref{system21}--\eqref{system23}. Accordingly, the relevant asymptotic region exists only when this system is solvable. Arguing as in \cite{Boutet2}, we obtain the following result.
\begin{lemma}\label{gexist}
	For $-E_1-\sqrt{2}E_2<\xi<-E_1$, the parameters $\mu,\ H_1,\ H_2$ can be uniquely determined as functions of $\xi$ from the system \eqref{xt2}.
\end{lemma}
\begin{proof}
	Since $\frac{\d \tilde{g}}{\d k}(k)=(\frac{\d \tilde{g}}{\d k}(k^*))^*$ for $k\in \C \backslash (\tilde{\Sigma}_1\cup\tilde{\Sigma}_2\cup \tilde{\Sigma}_{\mu})$, we have $\int\limits_{E}^{E^*}\d \tilde{g}(k)\in i\R $, so the condition \eqref{system23} is a real condition.
	Substituting \eqref{system21} and \eqref{system22} into \eqref{system23} yields an equation
	\begin{align}\label{eq5.6}
		\int\limits_{E_1-E_2 i}^{E_1+E_2 i}(k-E_1+\xi+H_1)\sqrt{\frac{(k-H_1)^2+E_2^2+2(\xi+H_1)(H_1-E_1)}{(k-E_1)^2+E^2_2}}\d k=0.
	\end{align}
	Introducing the new variables $u=\frac{H_1-E_1}{E_2},\ v=-\frac{\xi+E_1}{2E_2}$, $\tau=\frac{k-E_1}{iE_2}$, then the equation \eqref{eq5.6} reads
	\begin{align}\label{Huv}
		\mathcal{H}(u,v)=\int\limits_{-1}^{1}\sqrt{\frac{(i\tau+u)^2+2u(u-2v)+1}{1-\tau^2}}(i\tau+2v-u)\d \tau=0,
	\end{align}
	where $u>0$ and $0<v<\frac{\sqrt{2}}{2}$. One can verify that
	\begin{align}
		&\mathcal{H}(0,v)=4v>0,\ \mathcal{H}(+\infty,v)<0, \mathcal{H}(0,0)=\mathcal{H}(\frac{\sqrt{2}}{2},\frac{\sqrt{2}}{2})=0,\nonumber\\
		&\mathcal{H}_{u}(u,v)=(-6(u-v)^2+2v^2-1)\int\limits_{-1}^{1}\sqrt{\frac{1}{(1-\tau^2)((i\tau+u)^2+1-4uv+2u^2)}}\d \tau
	\end{align}
 Since
	\[
	-6(u-v)^2+2v^2-1<0,\qquad
	u>0,\qquad
	0<v<\frac{\sqrt{2}}{2},
	\]
	it remains to show that
	\[
	I(u,v):=\int\limits_{-1}^{1}
	\sqrt{\frac{1}{(1-\tau^2)\bigl((i\tau+u)^2+1-4uv+2u^2\bigr)}}\,d\tau
	\]
	is positive. Writing
	\[
	(i\tau+u)^2+1-4uv+2u^2
	=\tilde{A}(\tau)+i\tilde{B}(\tau),
	\]
	where
	\[
	\tilde{A}(\tau)=3u^2-4uv+1-\tau^2,\qquad
	\tilde{B}(\tau)=2u\tau,
	\]
	we have
	\[
	\tilde{A}(-\tau)=\tilde{A}(\tau),\qquad
	\tilde{B}(-\tau)=-\tilde{B}(\tau).
	\]
	Hence, choosing the continuous branch of the square root which is
	positive at $\tau=0$, we have
	\[
	\mathcal{F}(-\tau)=\overline{\mathcal{F}(\tau)},
	\qquad
	\mathcal{F}(\tau)
	:=
	\sqrt{
		\frac{1}{
			(1-\tau^2)\bigl(\tilde{A}(\tau)+i\tilde{B}(\tau)\bigr)
		}
	}.
	\]
	Therefore,
	\[
	I(u,v)
	=\int_{-1}^{1}\mathcal{F}(\tau)\,d\tau
	=2\int_{0}^{1}\Re\mathcal{F}(\tau)\,d\tau.
	\]
	Now
	\[
	\tilde{A}(0)
	=3u^2-4uv+1
	=3\left(u-\frac{2v}{3}\right)^2
	+1-\frac{4v^2}{3}>0,
	\]
	since $0<v<\frac{\sqrt{2}}{2}$. Thus $\mathcal{F}(0)>0$.
	Moreover, for $0<\tau<1$,
	\[
	\tilde{B}(\tau)=2u\tau>0,
	\]
	and hence
	\[
	0<
	\arg\bigl(\tilde{A}(\tau)+i\tilde{B}(\tau)\bigr)
	<\pi.
	\]
	Since $1-\tau^2>0$, the continuous branch chosen above satisfies
	\[
	-\frac{\pi}{2}
	<
	\arg\mathcal{F}(\tau)
	=
	-\frac{1}{2}
	\arg\bigl(\tilde{A}(\tau)+i\tilde{B}(\tau)\bigr)
	<0.
	\]
	Therefore,
	\[
	\Re\mathcal{F}(\tau)>0,
	\qquad 0<\tau<1,
	\]
	and consequently
	\[
	I(u,v)
	=2\int_{0}^{1}\Re\mathcal{F}(\tau)\,d\tau>0.
	\]
	It follows that
	\[
	\mathcal{H}_{u}(u,v)
	=
	\bigl(-6(u-v)^2+2v^2-1\bigr)I(u,v)<0,
	\qquad
	u>0,\quad 0<v<\frac{\sqrt{2}}{2}.
	\]
	Therefore, for every fixed
	$v\in(0,\frac{\sqrt{2}}{2})$, the function
	$\mathcal{H}(u,v)$ is strictly decreasing with respect to $u$.
	Together with
	\[
	\mathcal{H}(0,v)=4v>0,\qquad
	\mathcal{H}(+\infty,v)<0,
	\]
	equation \eqref{Huv} determines a unique function
	$u=u(v)$ for
	$v\in(0,\frac{\sqrt{2}}{2})$.
 Consequently, we have that the system \eqref{xt2} determines uniquely $H_1(\xi)$, $H_2(\xi)$, $\mu(\xi)$ for $\xi\in (-E_1-\sqrt{2}E_2, -E_1)$. Moreover, $H_1(-E_1)=E_1$ and $H_1(-E_1-\sqrt{2}E_2)=E_1+\frac{E_2}{\sqrt{2}}$.
\end{proof}
By construction of the $g$-function, when $\xi=-E_1$, the branch point $H(-E_1)=H_1(-E_1)+iH_2(-E_1)$ must coincide with the branch point $E=E_1+iE_2$. On the other hand, when $\xi=-E_1-\sqrt{2}E_2$, one has $H_2(-E_1-\sqrt{2}E_2)=0$ and $\mu(-E_1-\sqrt{2}E_2)=H_1(-E_1-\sqrt{2}E_2)=E_1+\frac{E_2}{\sqrt{2}}$.
Next we consider the distribution of the signs of $\Im (\tilde{g}(k))$. At $k=H$,
\begin{align}
	\tilde{g}(H)=\int\limits_{E}^{H}\d \tilde{g}(k)=(\int\limits_{E}^{H}+\int\limits_{E^*}^{H^*})(k-\mu)\frac{\sqrt{(k-H)(k-H^*)}}{\sqrt{(k-E)(k-E^*)}}\d k\in \R,
\end{align}
where we have used the symmetry $(\tilde{g}'(k))^{*}=\tilde{g}'(k^*)$. Since $\tilde{g}'(k)=(k-k_0)^{-\frac{1}{2}}$ as $k\to k_0,\ k_0=E,E^*$, $\tilde{g}'(k)=(k-k_1)^{\frac{1}{2}}$ as $k\to k_1,\ k_1=H,H^*$ and $\Im (\tilde{g}(E))=\Im (\tilde{g}(H))=\Im (\tilde{g}(H^*))=\Im (\tilde{g}(E^*))=0$, then the curve $\Im (\tilde{g}(k))=0$ must have three branches going out from $H$ and $H^*$, the points $E$ and $H$ are connected by the zero-level curve $\tilde{\Sigma}_{1}$ and the points $E^*$ and $H^*$ are connected by $\tilde{\Sigma}_{2}$, where $\Im (\tilde{g}(k))=0$ on $\tilde{\Sigma}_{1}\cup\tilde{\Sigma}_{2}$. Since $\tilde{g}(k)$ is real on the real axis and $\mu$ is a zero of $\tilde{g}'(k)$, then the curve $\Im (\tilde{g}(k))=0$ must have four branches going out from the point $\mu$, there exists a curve $\tilde{\Sigma}_{\mu}$ connecting $H^*$ and $H$, passing through $\mu$, where $\Im (\tilde{g}(k))=0$. Since $\tilde{g}(k)$ behaves like $\theta(k)$ for large $k$, then there exists a curve where $\Im (\tilde{g}(k))=0$, starting from $H$ and going to infinity along the asymptotic line $\Re k=-\xi$. These contours and signature table of the function $\Im (\tilde{g}(k))$ are depicted in Fig.\ref{signtable2}. Finally, once $\xi$ leaves $\hat{\xi}$, the point $F$ moves into the region where $\Im \tilde g$ is positive. To clarify this, we require the following lemma.
\begin{lemma}\label{lem:Im-gF-monotone}
	Let
	\[
	G(\xi):=\Im \tilde g(F;\xi),
	\qquad
	F:=H(\hat{\xi}),
	\]
	where $F$ is fixed. Then, for
	$
	\hat{\xi}<\xi<-E_1,
	$
	one has
	\begin{equation}\label{eq:Gprime}
		G'(\xi)
		=
		\Im\int_{H(\xi)}^{F}
		\partial_{\xi}\!\bigl(d\tilde g(k;\xi)\bigr)
		\geq0.
	\end{equation}
	Moreover,
	\begin{equation}\label{eq:Im-gF-strict}
		\Im\tilde g(F;\xi)>0,
		\qquad
		\hat{\xi}<\xi<-E_1.
	\end{equation}
\end{lemma}

\begin{proof}
	Since $H(\xi)$ is a branch point of the differential
	$d\tilde g(k;\xi)$ and, by construction,
	$
	\Im\tilde g(H(\xi);\xi)=0,
	$
	we have
	\begin{equation}\label{eq:G-Abel}
		G(\xi)
		=
		\Im\bigl(
		\tilde g(F;\xi)-\tilde g(H(\xi);\xi)
		\bigr)
		=
		\Im\int_{H(\xi)}^F d\tilde g(k;\xi).
	\end{equation}
	
	Differentiating \eqref{eq:G-Abel} with respect to $\xi$, we get
	\begin{equation}\label{eq:Gprime-again}
		G'(\xi)
		=
		\Im\int_{H(\xi)}^F
		\partial_{\xi}\!\bigl(d\tilde g(k;\xi)\bigr).
	\end{equation}
	
	We now identify the differential in
	\eqref{eq:Gprime-again}. Define
	\begin{equation}\label{eq:dp-xi}
		dp_\xi(k)
		:=
		\frac12\partial_{\xi}\!\bigl(d\tilde g(k;\xi)\bigr).
	\end{equation}
	It is a meromorphic differential on the genus-one curve
	\begin{equation}\label{eq:X-xi}
		\mathcal X_\xi:
		\qquad
		\tilde R^2(k;\xi)
		=
		(k-E)(k-E^*)
		(k-H(\xi))(k-H^*(\xi)).
	\end{equation}
	From the expansion
	$
	\frac{d\tilde g}{dk}(k;\xi)
	=
	2k+2\xi+\mathcal O(k^{-2}),
	\ 
	k\to\infty,
	$
	we obtain
	\begin{equation}\label{eq:dp-xi-infty}
		dp_\xi(k)
		=
		\bigl(1+\mathcal O(k^{-2})\bigr)\,dk,
		\qquad
		k\to\infty.
	\end{equation}
	
	Moreover, the normalization condition for $\tilde g$ gives
	$
	\int_{H(\xi)}^{H^*(\xi)}
	d\tilde g(k;\xi)=0.
	$
	Differentiating this identity with respect to $\xi$, and again using
	the vanishing of $d\tilde g$ at the moving endpoints, gives
	$
	\int_{H(\xi)}^{H^*(\xi)}dp_\xi(k)=0.
	$
	Therefore $dp_\xi$ is precisely the normalized quasi-momentum
	differential associated with the branch points
	$
	E, H(\xi), E^*, H^*(\xi).
	$
	
	Let
	\begin{equation}\label{eq:p-xi}
		p_\xi(k):=\int_{H(\xi)}^k dp_\xi.
	\end{equation}
	The sign properties of the normalized quasi-momentum imply that,
	in the upper half-plane,
	\begin{equation}\label{eq:p-xi-sign}
		\Im p_\xi(k)\geq0.
	\end{equation}
	Indeed, when the corresponding spectral bands do not intersect the
	real axis this follows from
	\cite[Lemma~2.2]{GravaJenkinsZhangZhang2026}, whereas the case in which the bands
	intersect the real axis follows from
	\cite[Lemma~3.2]{GJZZ2026b}.  Thus both possible spectral
	configurations of the auxiliary elliptic curve
	$\mathcal X_\xi$ are covered.
	
	Since $F\in\mathbb C^+$, \eqref{eq:p-xi-sign} yields
	$
	\Im p_\xi(F)\geq0.
	$
	Consequently,
	\begin{align}
		G'(\xi)
		&=
		\Im\int_{H(\xi)}^F
		\partial_{\xi}\!\bigl(d\tilde g(k;\xi)\bigr)
		\nonumber\\
		&=
		2\Im\int_{H(\xi)}^Fdp_\xi(k)
		=
		2\Im p_\xi(F)
		\geq0.
		\label{eq:G-nondecreasing}
	\end{align}
	Thus $G$ is nondecreasing. Since
	$
	F=H(\hat\xi),
	$
	we have
	$
	G(\hat\xi)
	=
	\Im\tilde g(H(\hat\xi);\hat\xi)
	=
	0,
	$
	and hence
	\begin{equation}\label{eq:G-nonnegative}
		G(\xi)\geq0,
		\qquad
		\hat\xi\leq\xi<-E_1.
	\end{equation}
	
	It remains to show that the inequality is strict in the interior.
	By the implicit-function construction of $H(\xi)$ and the
	nonvanishing of the corresponding derivative in
	Lemma~\ref{gexist}, the functions
	$
	H_1(\xi),\qquad H_2(\xi)
	$
	are real analytic for
	$
	-E_1-\sqrt{2}E_2<\xi<-E_1.
	$
	Consequently, since $F$ is fixed,
	$
	G(\xi)=\Im\tilde g(F;\xi)
	$
	is real analytic on $(\hat\xi,-E_1)$.
	
	We next examine the limit $\xi\to-E_1^-$. By
	Lemma~\ref{gexist},
	$
	H(\xi)\rightarrow E,
	\ 
	\mu(\xi)\rightarrow E_1.
	$
	Recalling that
	$
	d\tilde g(k;\xi)
	=
	2(k-\mu(\xi))
	\sqrt{
		\frac{(k-H(\xi))(k-H^*(\xi))}
		{(k-E)(k-E^*)}
	}\,dk,
	$
	we obtain
	$
	d\tilde g(k;\xi)
	\rightarrow
	2(k-E_1)\,dk,
	\ 
	\xi\to-E_1^-.
	$
	Therefore,
	\begin{align}
		\lim_{\xi\to-E_1^-}G(\xi)
		&=
		\Im\int_E^F2(k-E_1)\,dk
		\nonumber\\
		&=
		\Im\left[
		(F-E_1)^2-(E-E_1)^2
		\right]
		\nonumber\\
		&=
		2(F_1-E_1)F_2.
		\label{eq:G-positive-limit}
	\end{align}
	Since $F=H(\hat\xi)$ lies in the upper half-plane and, by
	Lemma~\ref{gexist},
	$
	F_1>E_1,
	\ 
	F_2>0,
	$
	it follows that
	\begin{equation}\label{eq:G-limit-positive}
		\lim_{\xi\to-E_1^-}G(\xi)>0.
	\end{equation}
	
	Suppose now, to the contrary, that there exists
	$
	\xi_0\in(\hat\xi,-E_1)
	$
	such that $G(\xi_0)=0$. Since $G$ is nondecreasing,
	$G(\hat\xi)=0$, and $G\geq0$, we must have
	$
	G(\xi)=0,
	\,
	\hat\xi<\xi<\xi_0.
	$
	Thus the real-analytic function $G$ vanishes on a nonempty open
	interval. By the identity theorem for real-analytic functions,
	$
	G(\xi)\equiv0,
	\hat\xi<\xi<-E_1.
	$
	This contradicts \eqref{eq:G-limit-positive}. Therefore no such
	$\xi_0$ exists, and hence
	$
	G(\xi)>0,
	\qquad
	\hat\xi<\xi<-E_1.
	$
	That is,
	$
	\Im\tilde g(F;\xi)>0,
	\hat\xi<\xi<-E_1.
	$
	The proof is complete.
\end{proof}
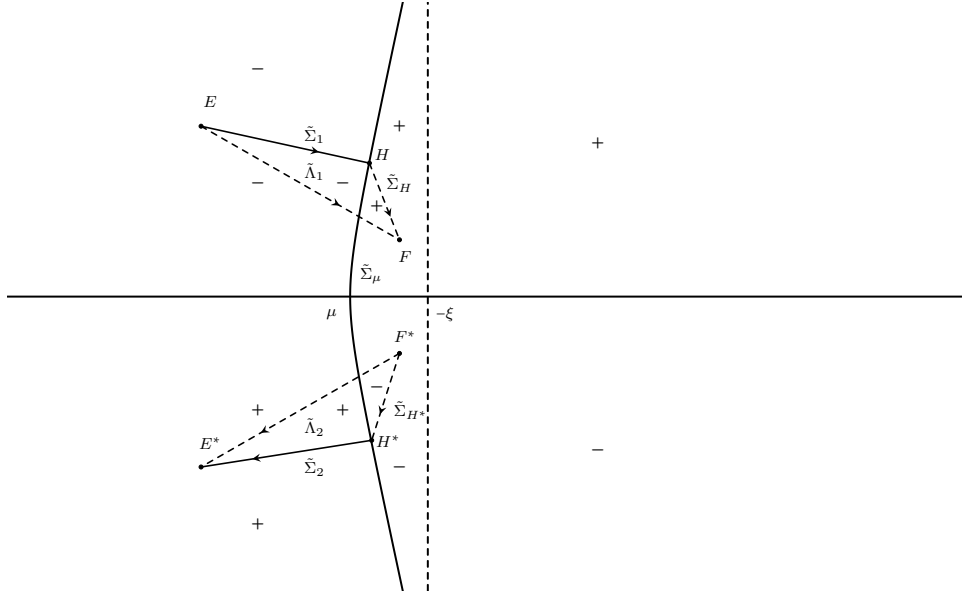
\begin{figure}
\centering
\scalebox{.75}{
\begin{tikzpicture}[line cap=round,line join=round,>=triangle 45,x=1cm,y=1cm,>=stealth]
\clip(-8.42,-5.18) rectangle (8.42,5.18);
\draw [samples=50,domain=-0.99:0.99,rotate around={0:(-2.6,0)},xshift=-2.6cm,yshift=0cm,line width=1pt] plot ({0.23*(1+(\x)^2)/(1-(\x)^2)},{1.05*2*(\x)/(1-(\x)^2)});
\draw[->- = .7,thick] (-5,3)-- (-2.03,2.35);
\draw [->- = .7,thick,dashed] (-5,3)-- (-1.5,1);
\draw [->- = .7,thick,dashed] (-2.03,2.35)--(-1.5,1);
\draw[->- = .7,thick] (-1.99,-2.53)--(-5,-3);
\draw [->- = .7,thick,dashed]  (-1.5,-1)--(-5,-3);
\draw [->- = .7,thick,dashed] (-1.5,-1)--(-1.99,-2.53);
\draw [line width=1pt,dashed] (-1,-5.18) -- (-1,5.18);
\draw [line width=1pt,domain=-8.42:8.42] plot(\x,{(-0-0*\x)/3.54});
\begin{scriptsize}
\draw [fill=black] (-5,3) circle (1pt);
\draw[color=black] (-4.84,3.43) node {$E$};
\draw[color=black] (-4,4) node {$\mathlarger{\mathlarger{\mathlarger{-}}}$};
\draw[color=black] (-4,2) node {$\mathlarger{\mathlarger{\mathlarger{-}}}$};
\draw[color=black] (-2.5,2) node {$\mathlarger{\mathlarger{\mathlarger{-}}}$};
\draw[color=black] (-4,-4) node {$\mathlarger{\mathlarger{\mathlarger{+}}}$};
\draw[color=black] (-4,-2) node {$\mathlarger{\mathlarger{\mathlarger{+}}}$};
\draw[color=black] (-2.5,-2) node {$\mathlarger{\mathlarger{\mathlarger{+}}}$};
\draw [fill=black] (-2.03,2.35) circle (1pt);
\draw[color=black] (-1.8,2.5) node {$H$};
\draw[color=black] (-1.5,2) node {$\tilde\Sigma_{H}$};
\draw[color=black] (-1.3,-2) node {$\tilde\Sigma_{H^*}$};
\draw [fill=black] (-1.5,1) circle (1pt);
\draw[color=black] (-1.4,0.7) node {$F$};
\draw[color=black] (-2,0.4) node {$\tilde{\Sigma}_{\mu}$};
\draw[color=black] (-1.5,3) node {$\mathlarger{\mathlarger{\mathlarger{+}}}$};
\draw[color=black] (2,2.7) node {$\mathlarger{\mathlarger{\mathlarger{+}}}$};
\draw[color=black] (-1.9,1.6) node {$\mathlarger{\mathlarger{\mathlarger{+}}}$};
\draw[color=black] (-1.5,-3) node {$\mathlarger{\mathlarger{\mathlarger{-}}}$};
\draw[color=black] (2,-2.7) node {$\mathlarger{\mathlarger{\mathlarger{-}}}$};
\draw[color=black] (-1.9,-1.6) node {$\mathlarger{\mathlarger{\mathlarger{-}}}$};
\draw [fill=black] (-5,-3) circle (1pt);
\draw[color=black] (-4.84,-2.57) node {$E^*$};
\draw [fill=black] (-1.99,-2.53) circle (1pt);
\draw[color=black] (-1.7,-2.54) node {$H^*$};
\draw [fill=black] (-1.5,-1) circle (1pt);
\draw[color=black] (-1.4,-0.7) node {$F^*$};
\draw[color=black] (-2.7,-0.3) node {$\mu$};
\draw[color=black] (-0.7,-0.3) node {$-\xi$};
\draw[color=black] (-3,2.85) node {$\tilde{\Sigma}_{1}$};
\draw[color=black] (-3,-3) node {$\tilde{\Sigma}_{2}$};
\draw[color=black] (-3,2.2) node {$\tilde{\Lambda}_{1}$};
\draw[color=black] (-3,-2.3) node {$\tilde{\Lambda}_{2}$};
\end{scriptsize}
\end{tikzpicture}
}
\caption{Signature table of $\Im (\tilde{g}(k))$, $\Im(\tilde{g}(k))=0$ on the solid line.}
\label{signtable2}
\end{figure}
Then the scalar function $\tilde{g}(k)$ satisfies the jump conditions:
\begin{align}
&\tilde{g}_+(k) + \tilde{g}_-(k) = 0,& k \in \tilde{\Sigma}_{1}\cup \tilde{\Sigma}_{2}, \label{c1}\\
&\tilde{g}_+(k)-\tilde{g}_-(k) =  \tilde \Omega,& k \in  \tilde{\Sigma}_{\mu}, \label{c2}
\end{align}
where
\begin{equation}
 \label{tOmega}
\tilde{\Omega}=4\int\limits_{E}^{H}(k-\mu)
\le(\sqrt{\frac{{(\zeta-H)(\zeta-H^*)}}{{(\zeta-E)(\zeta-E^*)}}}\ri)_{+}\d\zeta\in \R,
\end{equation}
By \eqref{gprime1},  we  obtain
\begin{align}\label{tginfty}
\tilde{g}_{\infty}=(\int\limits_{E}^{\infty}+\int\limits_{E^*}^{\infty})[(\zeta-\mu)\sqrt{\frac{{(\zeta-H)(\zeta-H^*)}}{{(\zeta-E)(\zeta-E^*)}}}-\xi-\zeta]\d\zeta+E^2_2-E^2_1-2E_1\xi\in \R.
\end{align}
Then we define a new matrix-valued function $\tilde{X}(k;x,t)$ by
\begin{align}
\tilde{X}(k;x,t)=X(k;x,t)\begin{cases}
\begin{pmatrix} {1}&0\\ \displaystyle {2ir(k) e^{2it\theta(k)} } & {1} \end{pmatrix}, & k \in \tilde{\Lambda}_1, \\
\begin{pmatrix} {1}&\displaystyle {2i\bar{r}(k) e^{ -2 it\theta(k)} }\\0 & {1} \end{pmatrix}, &  k \in \tilde{\Lambda}_2,\\
I,& elsewhere,
\end{cases}
\end{align}
where the domains $\tilde{\Lambda}_1$ and $\tilde{\Lambda}_2$ see Fig.\ref{signtable2}.
The result jumps for $\tilde{X}(k;x,t)$ are
\begin{align}
\tilde{X}_{+}(k;x,t)=\tilde{X}_{-}(k;x,t)\begin{cases}
\begin{pmatrix} {1}&0\\ \displaystyle {2ir(k) e^{2it\theta(k)} } & {1} \end{pmatrix}, & k \in \tilde{\Sigma}_1\cup \tilde{\Sigma}_{H}, \\
\begin{pmatrix} {1}& \displaystyle {2i\bar{r}(k) e^{ -2 it\theta(k)} }\\ 0 & {1} \end{pmatrix}, &  k \in \tilde{\Sigma}_2\cup \tilde{\Sigma}_{H^*}.
\end{cases}
\end{align}
where $\tilde{\Sigma}_{H}$ is an oriented contour connecting the point $H$ and $F$, and $\tilde{\Sigma}_{H^*}$ is an oriented contour connecting the point $F^*$ and $H^*$, see Fig.\ref{signtable2}.
Now we make the new transformation $\tilde{X}(k;x,t)\mapsto \tilde{T}(k;x,t)$  by
\begin{gather}
\tilde{T}(k;x,t) =\tilde{f}_{\infty}^{-\sigma_3}e^{-it\tilde{g}_{\infty}\sigma_3} \tilde{X}(k;x,t) e^{it(\tilde{g}(k)-\theta(k))\sigma_3}\tilde{f}(k)^{\sigma_3}.
\end{gather}
The jump conditions for $\tilde{T}(k;x,t)$ are given by
\begin{align}
\label{RH_T}
&\tilde{T}_+(k)=\tilde{T}_-(k)V_{\tilde{T}}(k),\\
&V_{\tilde{T}}(k)=\begin{cases}
\displaystyle \begin{pmatrix} e^{it\le(\tilde{g}_+(k) - \tilde{g}_-(k)\ri)}\dfrac{\tilde{f}_+(k)}{\tilde{f}_-(k)}&0\\  2i {r}(k)e^{it(\tilde{g}_{+}(k)+\tilde{g}_{-}(k))}\tilde{f}_+(k)\tilde{f}_-(k) & e^{-it\le(\tilde{g}_+(k) - \tilde{g}_-(k)\ri)}\dfrac{\tilde{f}_-(k)}{\tilde{f}_+(k)} \end{pmatrix}, &\quad k \in  \tilde{\Sigma}_1\cup\tilde{\Sigma}_{H},\\
\displaystyle \begin{pmatrix} e^{it\le(\tilde{g}_+(k) - \tilde{g}_-(k)\ri)}\dfrac{\tilde{f}_+(k)}{\tilde{f}_-(k)} & \dfrac{2i \bar{r}(k)e^{-it(\tilde{g}_{+}(k)+\tilde{g}_{-}(k))}}{\tilde{f}_+(k)\tilde{f}_-(k)}\\0  & e^{-it\le(\tilde{g}_+(k) - \tilde{g}_-(k)\ri)}\dfrac{\tilde{f}_-(k)}{\tilde{f}_+(k)} \end{pmatrix}, & \quad k \in \tilde{\Sigma}_2\cup \tilde{\Sigma}_{H^*}, \\
\displaystyle \begin{pmatrix} e^{it\tilde{\Omega} }\dfrac{\tilde{f}_+(k)}{\tilde{f}_-(k)} & 0 \\0& e^{-it \tilde{\Omega}}\dfrac{\tilde{f}_-(k)}{\tilde{f}_+(k)} \end{pmatrix}, &\quad k \in  \tilde{\Sigma}_{\mu}.
\end{cases}
 \end{align}
Then we choose the function $\tilde{f}(k)$ to simplify the jumps on $\tilde{\Sigma}_{1}$ and $\tilde{\Sigma}_{2}$ via
 \begin{align}
& \tilde{f}_+(k) \tilde{f}_-(k) =\frac{1}{2r(k)}, & k \in \tilde{\Sigma}_{1}, \label{Tfcon1} \\
&\tilde{f}_+(k) \tilde{f}_-(k) ={2\bar{r}(k)}, & k \in \tilde{\Sigma}_{2} ,\\
& \frac{\tilde{f}_+(k)}{ \tilde{f}_-(k)} =e^{i\tilde{\Delta}} ,& k \in \tilde{\Sigma}_{\mu}, \\
& \tilde{f}(k) = \tilde{f}_{\infty}+\mathcal{O}\le(\frac{1}{k}\ri) ,& k \to  \infty  \ .\label{Tfcon3}
\end{align}
It is easy to check that the function $\tilde{f}(k)$ is given by
\begin{equation}\label{tf}
\tilde{f}(k)=\exp\le\{ \frac{\tilde{R}(k)}{2\pi i}\left[\int_{ \tilde{\Sigma}_{1}}  \frac{\log  (\frac{1}{2r(\zeta)}) }{\tilde{R}_{+}(\zeta)(\zeta-k)} \d \zeta+
\int_{\tilde{\Sigma}_{2}}  \frac{\log (2\bar{r}(\zeta ))}{\tilde{R}_{+}(\zeta)(\zeta-k)} \d \zeta+\int^{H^*}_{H}  \frac{i\tilde{\Delta}}{\tilde{R}(\zeta)(\zeta-k)} \d \zeta
\right] \ri\} \, ,
\end{equation}
where $\tilde{R}(k)=\sqrt{(k-E)(k-E^*)(k-H)(k-H^*)}$ and the condition  \eqref{Tfcon3} determines $\tilde{\Delta}$ as
\begin{align}\label{tDelta}
\tilde{\Delta}&=i\left[\int_{\tilde{\Sigma}_{1}}  \frac{\log (\frac{1}{2r(\zeta)})}{\tilde{R}_{+}(\zeta)} \d \zeta+
\int_{\tilde{\Sigma}_{2}}  \frac{\log (2\bar{r}(\zeta))}{\tilde{R}_{+}(\zeta)} \d \zeta\right]\left[\int^{H^*}_{H}  \frac{ \d \zeta}{\tilde{R}(\zeta)}\right]^{-1},
\end{align}
and
\begin{align}\label{tfinfty}
\tilde{f}_{\infty}=\exp\{ -\frac{1}{2\pi i}[\int_{ \tilde{\Sigma}_{1}}  \frac{\log  (\frac{1}{2r(\zeta)})(\zeta-e_2) }{\tilde{R}_{+}(\zeta)} \d \zeta+
\int_{\tilde{\Sigma}_{2}}  \frac{\log (2\bar{r}(\zeta ))(\zeta-e_2)}{\tilde{R}_{+}(\zeta)} \d \zeta+\int^{H^*}_{H}  \frac{i\tilde{\Delta}(\zeta-e_2)}{\tilde{R}(\zeta)} \d \zeta
]\},
\end{align}
with $e_2=\frac{E+E^*+H+H^*}{2}$.
Since $\tilde{g}(k)$ and $\tilde{f}(k)$ do not have jumps across $\tilde{\Sigma}_{H}\cup \tilde{\Sigma}_{H^*}$, then $\tilde{T}(k;x,t)$ satisfies the following RH problem.
\begin{RHP}\label{RHP6} Find a matrix-valued function $\tilde{T}(k;x,t)$ with the following properties
\begin{enumerate}
\item  $\tilde{T}(k;x,t)$  is analytic for $k\in\mathbb{C}\backslash (\tilde{\Sigma}_{1}\cup\tilde{\Sigma}_{2}\cup \tilde{\Sigma}_{\mu}\cup \tilde{\Sigma}_{H}\cup \tilde{\Sigma}_{H^*})$, where the jump contours are shown in Fig.\ref{signtable2}.
\item  For $k\in \tilde{\Sigma}_{1}\cup\tilde{\Sigma}_{2}\cup \tilde{\Sigma}_{\mu}\cup \tilde{\Sigma}_{H}\cup \tilde{\Sigma}_{H^*}$, the boundary values $\tilde{T}_{\pm}(k;x,t)$ satisfy the following jump relation
\begin{align}
\tilde{T}_+(k) = \tilde{T}_-(k)  \begin{cases}
\displaystyle \begin{pmatrix} e^{it\le(\tilde{g}_+(k) - \tilde{g}_-(k)\ri)} \frac{\tilde{f}_{+}(k)}{\tilde{f}_{-}(k)}& 0 \\ i& e^{-it\le(\tilde{g}_+(k) - \tilde{g}_-(k)\ri)} \frac{\tilde{f}_{-}(k)}{\tilde{f}_{+}(k)}\end{pmatrix}, &\quad k \in  \tilde{\Sigma}_{1},\\
\displaystyle \begin{pmatrix} e^{it\le(\tilde{g}_+(k) - \tilde{g}_-(k)\ri)}  \frac{\tilde{f}_{+}(k)}{\tilde{f}_{-}(k)}&i \\ 0& e^{-it\le(\tilde{g}_+(k) - \tilde{g}_-(k)\ri)} \frac{\tilde{f}_{-}(k)}{\tilde{f}_{+}(k)} \end{pmatrix}, & \quad k \in \tilde{\Sigma}_{2}, \\
\displaystyle \begin{pmatrix} 1  & 0\\   2i {r}(k)\tilde{f}^2(k)e^{2it\tilde{g}(k)} & 1 \end{pmatrix}, &\quad k \in  \tilde{\Sigma}_{H},\\
\displaystyle \begin{pmatrix} 1 & \frac{2i\bar{r}(k)}{\tilde{f}^2(k)} e^{-2it \tilde{g}(k)}\\0 & 1  \end{pmatrix}, & \quad k \in \tilde{\Sigma}_{H^*}, \\
\displaystyle\begin{pmatrix} e^{i(\tilde{\Omega} t + \tilde{ \Delta} ) } &0 \\ 0 & e^{-i(\tilde \Omega t + \tilde{ \Delta} ) }  \end{pmatrix}, &\quad k \in   \tilde{\Sigma}_{\mu}.
 \end{cases}\end{align}
\item
$
 \tilde{T}(k) = I + \mathcal{O}\le(\frac{1}{k}\ri), \qquad k \to \infty.
$
\end{enumerate}
\end{RHP}
\subsubsection{Opening lenses}
The jump matrix of $\tilde{T}(k;x,t)$ on $\tilde{\Sigma}_{1}$  can be decomposed into
\begin{align*}
& \begin{pmatrix} e^{it\le(\tilde{g}_+(k) - \tilde{g}_-(k)\ri)}\dfrac{\tilde{f}_+(k)}{\tilde{f}_-(k)}&0 \\ i & e^{-it\le(\tilde{g}_+(k) - \tilde{g}_-(k)\ri)}\dfrac{\tilde{f}_-(k)}{\tilde{f}_+(k)} \end{pmatrix}\\
&\quad\quad=
\begin{pmatrix}1& -i\dfrac{e^{-2it\tilde{g}_{-}(k)}}{2\hat{r}(k)\tilde{f}_-^2 (k)}\\ 0 &1 \end{pmatrix}
\begin{pmatrix}0&i\\i&0\end{pmatrix}
\begin{pmatrix}1&-i\dfrac{e^{-2it\tilde{g}_{+}(k)}}{2\hat{r}(k)\tilde{f}_+^2(k)}\\0&1 \end{pmatrix},
\end{align*}
and on $\tilde{\Sigma}_{2}$ as
\begin{align*}
&\begin{pmatrix} e^{it\le(\tilde{g}_+(k) - \tilde{g}_-(k)\ri)}\dfrac{\tilde{f}_+(k)}{\tilde{f}_-(k)} & i \\ 0  & e^{-it\le(\tilde{g}_+(k) - \tilde{g}_-(k)\ri)}\dfrac{\tilde{f}_-(k)}{\tilde{f}_+(k)} \end{pmatrix}\\
&\quad\quad=
\begin{pmatrix}1&0\\-i\dfrac{e^{2it\tilde{g}_{-}(k)}\tilde{f}_-^2 (k)}{2\hat{r}(k)}& 1 \end{pmatrix}
\begin{pmatrix}0&i \\ i &0\end{pmatrix}
\begin{pmatrix}1&0\\-i\dfrac{e^{2it \tilde{g}_{+}(k)}\tilde{f}^2_+(k)}{2\hat{r}(k)}& 1 \end{pmatrix},
\end{align*}
where   $\hat{r}(k)$  is the analytic  continuation of the function $r(k)$ and $\bar{r}(k)$  off the contours $\tilde{\Sigma}_{1}\cup \tilde{\Sigma}_{2}$  and satisfy the following condition:
$
\hat{r}(k)= r(k), \ k\in \tilde{\Sigma}_{1},\ \ \hat{r}(k)=\bar{r}(k), \ k\in \tilde{\Sigma}_{2}.
$
With the help of these factorizations, we can open lenses $\tilde{\mathcal{C}}_j$ around each $\tilde{\Sigma}_{j}$, $j=1,2$.  We use these lenses to define the transformation
\begin{equation}
	\tilde{S}(k;x,t) = \begin{dcases}
		\tilde{T}(k;x,t)\begin{pmatrix} 1 &  i\frac{e^{-2it\tilde{g}(k)}} {2\hat{r}(k)\tilde{f}^2(k) }  \\0 & 1 \end{pmatrix}
		& k\in\text{ the lens above $\tilde{\Sigma}_1$ } ,\\
\tilde{T}(k;x,t) \begin{pmatrix} 1 &  -i\frac{e^{-2it\tilde{g}(k)}} {2\hat{r}(k)\tilde{f}^2(k) }  \\0 & 1 \end{pmatrix}
		& k\in\text{ the lens below $\tilde{\Sigma}_1$ } ,\\
		\tilde{T}(k;x,t) \begin{pmatrix} 1 &0\\ i\frac{\tilde{f}^2(k)} {2\hat{r}(k)}  e^{2it\tilde{g}(k) }&1 \end{pmatrix}
		& k\in \text{ the lens below $\tilde{\Sigma}_2$} ,\\
\tilde{T}(k;x,t) \begin{pmatrix} 1 &0\\ -i\frac{\tilde{f}^2(k)} {2\hat{r}(k)}  e^{2it\tilde{g}(k) }&1 \end{pmatrix}
		& k\in \text{ the lens above $\tilde{\Sigma}_2$} ,\\
		\tilde{T}(k;x,t)
		& \text{elsewhere},
	\end{dcases}
\end{equation}
the lenses are shown in Fig.\ref{openlenseswtS}.
The matrix-valued function $\tilde{S}(k;x,t)$ satisfies the following RH problem.
\begin{RHP}\label{RHP7} Find a matrix-valued function $\tilde{S}(k;x,t)$ with the following properties
\begin{enumerate}
\item $\tilde{S}(k;x,t)$ is analytic for $k\in\C\backslash \le(\tilde{\mathcal {C}}_1\cup\tilde{{\mathcal C}}_2\cup \tilde{\Sigma}_{1}\cup\tilde{\Sigma}_{2}\cup\tilde{\Sigma}_{H}\cup\tilde{\Sigma}_{H^*} \cup\tilde{\Sigma}_{\mu}\ri)$ and the contours are shown in Fig.\ref{openlenseswtS}.
\item For $k\in \tilde{\mathcal {C}}_1\cup\tilde{{\mathcal C}}_2\cup \tilde{\Sigma}_{1}\cup\tilde{\Sigma}_{2}\cup\tilde{\Sigma}_{H}\cup\tilde{\Sigma}_{H^*} \cup\tilde{\Sigma}_{\mu}$, the boundary values $\tilde{S}_{\pm}(k)$ satisfy the jump relation
\begin{equation}
\tilde{S}_+(k)=\tilde{S}_-(k)\begin{cases}
\displaystyle \begin{pmatrix} 1 &  -\frac{ie^{-2it\tilde{g}(k)}} {2\hat{r}(k)\tilde{f}^2(k) }  \\0 & 1 \end{pmatrix} , &\quad k \in  \tilde{\mathcal{C}}_{1},\\
\displaystyle \begin{pmatrix} 1 &0\\ -\frac{i\tilde{f}^2(k)} {2\hat{r}(k)}  e^{2it\tilde{g}(k) }&1 \end{pmatrix}, & \quad k \in \tilde{\mathcal{C}}_{2}, \\
\displaystyle \begin{pmatrix} 1  & 0\\  { 2i {r}(k) }{\tilde{f}_{+}(k)\tilde{f}_{-}(k)}e^{2it\tilde{g}(k)} & 1 \end{pmatrix}, &\quad k \in  \tilde{\Sigma}_{H},\\
\displaystyle \begin{pmatrix} 1 & \frac{2i\bar{r}(k)}{\tilde{f}_{+}(k)\tilde{f}_{-}(k)} e^{-2it \tilde{g}(k)}\\0 & 1  \end{pmatrix}, & \quad k \in \tilde{\Sigma}_{H^*}, \\
\displaystyle \begin{pmatrix} 0  & i\\i & 0 \end{pmatrix}, & \quad k \in \tilde{\Sigma}_{1}\cup\tilde{\Sigma}_{2}, \\
\displaystyle\begin{pmatrix} e^{i(\tilde{\Omega} t + \tilde{ \Delta} ) } &0 \\ 0 & e^{-i(\tilde \Omega t + \tilde{ \Delta} ) }  \end{pmatrix}, &\quad k \in   \tilde{\Sigma}_{\mu}.
 \end{cases}
\end{equation}
\item
$
\tilde{S}(k)=I + \mathcal{O}\le(\frac{1}{k}\ri), \quad k \to \infty.
$
\end{enumerate}
\end{RHP}

\begin{figure}
\centering
\scalebox{.75}{
\begin{tikzpicture}[line cap=round,line join=round,>=triangle 45,x=1cm,y=1cm,>=stealth]
\clip(-8.42,-5.18) rectangle (8.42,5.18);
\draw [samples=50,domain=-0.99:0.99,rotate around={0:(-2.6,0)},xshift=-2.6cm,yshift=0cm,line width=1pt] plot ({0.23*(1+(\x)^2)/(1-(\x)^2)},{1.05*2*(\x)/(1-(\x)^2)});
\draw[->- = .7,thick] (-5,3)-- (-2,2.4);
\draw[->- = .7,thick](-2,-2.4)--(-5,-3);
\draw [line width=1pt,dashed] (-1,-5.18) -- (-1,5.18);
\draw [line width=1pt,domain=-8.42:8.42] plot(\x,{(-0-0*\x)/3.54});
\draw[->- = .7,thick] (-5,3)-- (-3.3,3.589);
\draw[->- = .7,thick]  (-3.3,3.589)--(-2,2.4);
\draw[->- = .7,thick] (-5,3)-- (-3.7,1.76);
\draw[->- = .7,thick]  (-3.7,1.76)--(-2,2.4);

\draw[->- = .7,thick] (-2,-2.4)-- (-3.3,-3.589);
\draw[->- = .7,thick]  (-3.3,-3.589)--(-5,-3);
\draw[->- = .7,thick] (-2,-2.4)-- (-3.7,-1.76);
\draw[->- = .7,thick]  (-3.7,-1.76)--(-5,-3);
\draw[color=black] (-2,0.4) node {$\tilde{\Sigma}_{\mu}$};
\begin{scriptsize}
\draw [fill=black] (-5,3) circle (1pt);
\draw[color=black] (-5.2,2.9) node {$E$};
\draw [fill=black] (-2,2.4) circle (1pt);
\draw[color=black] (-1.8,2.5) node {$H$};
\draw [fill=black] (-1.4,1.2) circle (1pt);
\draw[color=black] (-1.4,1) node {$F$};
\draw[color=black] (-1.4,2) node {$\tilde{\Sigma}_{H}$};
\draw[->- = .7,thick,dashed]   (-2,2.4)--(-1.4,1.2);
\draw [fill=black] (-5,-3) circle (1pt);
\draw[color=black] (-5.2,-2.9) node {$E^*$};
\draw [fill=black] (-2,-2.4) circle (1pt);
\draw[color=black] (-1.7,-2.54) node {$H^*$};
\draw [fill=black] (-1.4,-1.2) circle (1pt);
\draw[color=black] (-1.4,-1) node {$F^*$};
\draw[color=black] (-1.4,-2) node {$\tilde{\Sigma}_{H^*}$};
\draw[->- = .7,thick,dashed]   (-2,-2.4)--(-1.4,-1.2);

\draw[color=black] (-2.7,-0.3) node {$\mu$};
\draw[color=black] (-0.7,-0.3) node {$-\xi$};
\draw[color=black] (-3,2.85) node {$\tilde{\Sigma}_{1}$};
\draw[color=black] (-3,-2.9) node {$\tilde{\Sigma}_{2}$};
\draw[color=black] (-4,3.7) node {$\tilde{\mathcal{C}}_{1}$};
\draw[color=black] (-4.5,2.2) node {$\tilde{\mathcal{C}}_{1}$};
\draw[color=black] (-4,-3.7) node {$\tilde{\mathcal{C}}_{2}$};
\draw[color=black] (-4.5,-2.2) node {$\tilde{\mathcal{C}}_{2}$};
\end{scriptsize}
\end{tikzpicture}
}
\caption{Opening lenses.  }
\label{openlenseswtS}
\end{figure}
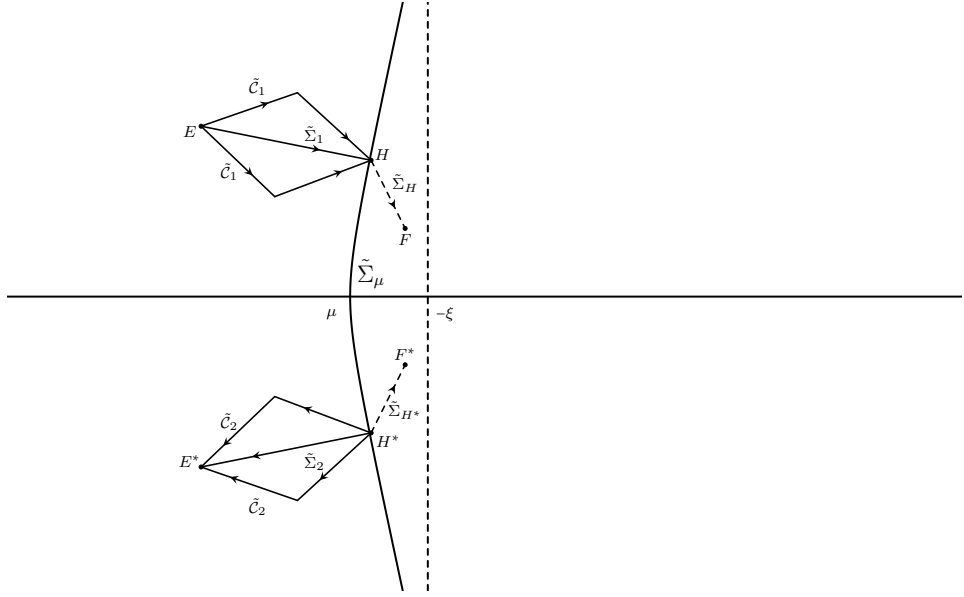
From the signature table of the function $\Im (\tilde g(k))$ as shown in Fig.\ref{signtable2}, we find that the following inequalities hold:
\begin{eqnarray}
&&\Re [2 i\tilde{g}(k)]>0 \ \mbox{ for }k \in \tilde{\mathcal{C}}_{1} \backslash \{E,H \} \ ,
\label{teq1}\\
&&\Re [2i \tilde{g}(k)] <0 \ \mbox{ for }k \in \tilde{\mathcal{C}}_{2} \backslash \{ E^*,H^* \} \ , \label{teq2}\\
&&\Re [2i\tilde{g}(k)]< 0 \mbox{ for } k \in \tilde{\Sigma}_{H} \ , \label{teq3}\\
&& \Re [2i\tilde{g}(k)]> 0 \mbox{ for } k \in \tilde{\Sigma}_{H^*}\label{teq4}.
\end{eqnarray}
By above inequalities, the jump matrices of $\tilde{S}(k;x,t)$ exhibit exponential convergence towards constant jumps in regions exterior to the vicinities of $E$, $E^*$, $H$ and $H^*$ as $t\to +\infty$. We then obtain the model RH problem for $\tilde{S}^{\infty}(k;x,t)$.
\begin{RHP}\label{RHP8} Find a matrix-valued function $\tilde{S}^{\infty}(k;x,t)$ with the following properties
\begin{enumerate}
\item $\tilde{S}^{\infty}(k;x,t)$ is analytic for $k\in\C \backslash \le( \tilde{\Sigma}_{1}\cup \tilde{\Sigma}_{2}\cup\tilde{\Sigma}_{\mu}\ri)$.
\item For $\tilde{\Sigma}_{1}\cup \tilde{\Sigma}_{2}\cup\tilde{\Sigma}_{\mu}$, the boundary values $\tilde{S}^{\infty}_{\pm}(k;x,t)$ satisfy the jump relation
\begin{gather}
\label{Stinfinity1}
\tilde{S}^{\infty}_+(k) = \tilde{S}^{\infty}_-(k)
\begin{cases}
\begin{pmatrix} e^{i(t\tilde{\Omega}+\tilde{ \Delta}) } &0 \\ 0 & e^{-i(t\tilde{\Omega}+\tilde{ \Delta}) }\end{pmatrix},  & k \in \tilde{\Sigma}_{\mu},  \\
 \begin{pmatrix}0 & i \\ i & 0 \end{pmatrix}, &k \in \tilde{\Sigma}_{1}\cup\tilde{\Sigma}_{2},
  \end{cases}
\end{gather}
\item
$
\tilde{S}^{\infty}(k)=I+\mathcal{O}\le(\frac{1}{k}\ri),\quad k\to\infty.
$
\end{enumerate}
\end{RHP}
In order to obtain the solution of the model RH problem for $\tilde{S}^{\infty}(k;x,t)$, we start by introducing the holomorphic differential
\begin{eqnarray}
\label{HD}
\tilde{\omega}=\left(2\int_{H}^{H^*}\frac{\mathrm{d}k}{\tilde{R}(k)}\right)^{-1}\frac{\mathrm{d}k}{\tilde{R}(k)} \, ,
\end{eqnarray}
where $\tilde{R}(k)=\sqrt{(k-E)(k-E^*)(k-H)(k-H^*)}$ which has the asymptotics $\tilde{R}(k)=k^2(1+\mathcal{O}(\frac{1}{k}))$ as $k\to \infty$.
Utilizing equation \eqref{HD}, we introduce the integral
\begin{equation}\label{tAk}
\tilde{A}(k)=\int_{E}^{k}\tilde{\omega},\quad k\in\mathbb{C}\backslash \le( \tilde{\Sigma}_{1}\cup\tilde{\Sigma}_{2}\cup\tilde{\Sigma}_{\mu} \ri),
\end{equation}
and define the period $$\tilde{\tau}=\left(\int_{H}^{H^*}\frac{\mathrm{d}k}{\tilde{R}(k)}\right)^{-1}\int\limits_{E}^{H}\frac{\mathrm{d}k}{\tilde{R}(k)}.$$
Then we define the function $\tilde{\gamma}(k)$
\begin{align}
\tilde{\gamma}(k)=\le(\dfrac{(k-H^*)(k-E^*)}{(k-H)(k-E)}\ri)^{\frac{1}{4}},
\end{align}
and the function $\tilde{\gamma}(k)+\frac{1}{\tilde{\gamma}(k)}$ has a zero at $\tilde{E}_{0}=\frac{EH-E^*H^*}{E+H-E^*-H^*}$.
By the methodology established in section \ref{sec:3}, we get the solution of the  RH problem \ref{RHP8}:
\begin{equation}
\label{TildeSinfty_sol}
\begin{split}
\tilde{S}^{\infty}_{11}(k;x,t)=\frac{1}{2}(\tilde\gamma(k)+\frac{1}{\tilde\gamma(k)})\dfrac{\vartheta_3(\tilde{A}(k)-\tilde{A}(\tilde{E}_{0})-\frac{\tilde{\tau}}{2}-\frac{t \tilde{\Omega}+\tilde{\Delta}}{2\pi };\tilde{\tau})}{\vartheta_3(\tilde{A}(k)-\tilde{A}(\tilde{E}_{0})-\frac{\tilde{\tau}}{2}-\frac{1}{2};\tilde{\tau})}
\dfrac{\vartheta_3(\tilde{A}(\infty)-\tilde{A}(\tilde{E}_{0})-\frac{\tilde{\tau}}{2}-\frac{1}{2};\tilde{\tau})}{\vartheta_3(\tilde{A}(\infty)-\tilde{A}(\tilde{E}_{0})-\frac{\tilde{\tau}}{2}-\frac{t \tilde{\Omega}+\tilde{\Delta}}{2\pi };\tilde{\tau})},\\
\tilde{S}^{\infty}_{12}(k;x,t)=\frac{1}{2}(\tilde\gamma(k)-\frac{1}{\tilde\gamma(k)})\dfrac{\vartheta_3(\tilde{A}(k)+\tilde{A}(\tilde{E}_{0})+\frac{\tilde{\tau}}{2}+\frac{t \tilde{\Omega}+\tilde{\Delta}}{2\pi };\tilde{\tau})}{\vartheta_3(\tilde{A}(k)+\tilde{A}(\tilde{E}_{0})+\frac{\tilde{\tau}}{2}+\frac{1}{2};\tilde{\tau})}
\dfrac{\vartheta_3(\tilde{A}(\infty)-\tilde{A}(\tilde{E}_{0})-\frac{\tilde{\tau}}{2}-\frac{1}{2};\tilde{\tau})}{\vartheta_3(\tilde{A}(\infty)-\tilde{A}(\tilde{E}_{0})-\frac{\tilde{\tau}}{2}-\frac{t \tilde{\Omega}+\tilde{\Delta}}{2\pi };\tilde{\tau})},\\
\tilde{S}^{\infty}_{21}(k;x,t)=\frac{1}{2}(\tilde\gamma(k)-\frac{1}{\tilde\gamma(k)})\dfrac{\vartheta_3(\tilde{A}(k)+\tilde{A}(\tilde{E}_{0})+\frac{\tilde{\tau}}{2}-\frac{t \tilde{\Omega}+\tilde{\Delta}}{2\pi };\tilde{\tau})}{\vartheta_3(\tilde{A}(k)+\tilde{A}(\tilde{E}_{0})+\frac{\tilde{\tau}}{2}+\frac{1}{2};\tilde{\tau})}
\dfrac{\vartheta_3(\tilde{A}(\infty)-\tilde{A}(\tilde{E}_{0})-\frac{\tilde{\tau}}{2}-\frac{1}{2};\tilde{\tau})}{\vartheta_3(\tilde{A}(\infty)-\tilde{A}(\tilde{E}_{0})-\frac{\tilde{\tau}}{2}+\frac{t \tilde{\Omega}+\tilde{\Delta}}{2\pi };\tilde{\tau})},\\
\tilde{S}^{\infty}_{22}(k;x,t)=\frac{1}{2}(\tilde\gamma(k)+\frac{1}{\tilde\gamma(k)})\dfrac{\vartheta_3(\tilde{A}(k)-\tilde{A}(\tilde{E}_{0})-\frac{\tilde{\tau}}{2}+\frac{t \tilde{\Omega}+\tilde{\Delta}}{2\pi };\tilde{\tau})}{\vartheta_3(\tilde{A}(k)-\tilde{A}(\tilde{E}_{0})-\frac{\tilde{\tau}}{2}-\frac{1}{2};\tilde{\tau})}
\dfrac{\vartheta_3(\tilde{A}(\infty)-\tilde{A}(\tilde{E}_{0})-\frac{\tilde{\tau}}{2}-\frac{1}{2};\tilde{\tau})}{\vartheta_3(\tilde{A}(\infty)-\tilde{A}(\tilde{E}_{0})-\frac{\tilde{\tau}}{2}+\frac{t \tilde{\Omega}+\tilde{\Delta}}{2\pi };\tilde{\tau})}.
\end{split}
\end{equation}
We now construct the local parametrices at the four branch points.  From \eqref{gprime1},
\begin{align}\label{tg-local-behavior}
 d\tilde g(k)&=\mathcal O((k-E)^{-\frac12})\,dk,\qquad k\to E,\nonumber\\
 d\tilde g(k)&=\mathcal O((k-E^*)^{-\frac12})\,dk,\qquad k\to E^*,\\
 d\tilde g(k)&=\hat{c}_H(k-H)^{\frac12}(1+\mathcal O(k-H))\,dk,\qquad k\to H,\nonumber\\
 d\tilde g(k)&=\hat{c}_{H^*}(k-H^*)^{\frac12}(1+\mathcal O(k-H^*))\,dk,\qquad k\to H^*,\nonumber
\end{align}
where $\hat{c}_H,\hat{c}_{H^*}\neq0$.  Consequently,
\begin{align}\label{tg-32}
 \tilde g(k)-\tilde g(H)&=\frac{2\hat{c}_H}{3}(k-H)^{\frac32}(1+\mathcal O(k-H)),\qquad k\to H,\\
 \tilde g(k)-\tilde g(H^*)&=\frac{2\hat{c}_{H^*}}{3}(k-H^*)^{\frac32}(1+\mathcal O(k-H^*)),\qquad k\to H^*.
\end{align}
Thus the local models at $H,H^*$ are of Airy type.  In contrast, $\tilde g(k)-\tilde g(E)=\mathcal O((k-E)^{1/2})$ and $\tilde g(k)-\tilde g(E^*)=\mathcal O((k-E^*)^{1/2})$, so the local models at $E,E^*$ are of modified-Bessel type, exactly as in Section~\ref{TtoS}.

Fix mutually disjoint disks $B_\rho^H,B_\rho^{H^*},B_\rho^E,B_\rho^{E^*}$ whose radii are independent of $t$ and small enough that no other branch point or zero of the analytic factors lies in their closures.  In $B_\rho^H$ define
\begin{equation}\label{phiH}
 \phi_H(k):=i\bigl(\tilde g(k)-\tilde g(H)\bigr),
 \qquad
 \zeta_H(k):=\left(\frac32t\phi_H(k)\right)^{\frac23}.
\end{equation}
The branch is chosen so that $\zeta_H(H)=0$ and the band and the two lips of the lens are mapped to the standard Airy rays.  By \eqref{tg-32},
\begin{equation}\label{zetaH-conformal}
 \zeta_H(k)=t^{\frac23}\zeta_{H,0}(k),\qquad
 \zeta_{H,0}(H)=0,\qquad \zeta_{H,0}'(H)\neq0,
\end{equation}
so $\zeta_{H,0}$ is conformal in $B_\rho^H$.  Moreover,
\begin{equation}\label{Airy-phase-id}
 \frac23\zeta_H(k)^{\frac32}=t\phi_H(k)
 =it\bigl(\tilde g(k)-\tilde g(H)\bigr).
\end{equation}

Let $\Psi_{\rm Ai}(\zeta)$ denote the standard Airy model matrix.  It is constructed from
$\operatorname{Ai}(\zeta)$, $\operatorname{Ai}(\omega\zeta)$ and $\operatorname{Ai}(\omega^2\zeta)$, where $\omega=e^{2\pi i/3}$, and is analytic in the complement of the four rays
$\arg\zeta=0,\pm2\pi/3,\pi$.  With the orientation inherited from the local contour, its jumps are the constant Airy jumps
\begin{equation}\label{Airy-jumps}
 \Psi_{{\rm Ai},+}=\Psi_{{\rm Ai},-}
 \begin{cases}
 \begin{pmatrix}1&1\\0&1\end{pmatrix},&\arg\zeta=0,\\[2mm]
 \begin{pmatrix}1&0\\1&1\end{pmatrix},&\arg\zeta=\pm2\pi/3,\\[2mm]
 \begin{pmatrix}0&1\\-1&0\end{pmatrix},&\arg\zeta=\pi,
 \end{cases}
\end{equation}
where reversing the orientation of a ray replaces the corresponding matrix by its inverse.  As $\zeta\to\infty$ away from the rays,
\begin{equation}\label{Airy-asymptotics}
 \Psi_{\rm Ai}(\zeta)
 =\zeta^{-\frac14\sigma_3}M_{\rm Ai}
 \left(I+\frac{\Psi_1^{\rm Ai}}{\zeta^{3/2}}
 +\mathcal O(\zeta^{-3})\right)
 e^{-\frac23\zeta^{3/2}\sigma_3},
 \qquad
 M_{\rm Ai}:=\frac1{\sqrt2}\begin{pmatrix}1&i\\ i&1\end{pmatrix}.
\end{equation}
The precise constant matrix $\Psi_1^{\rm Ai}$ is immaterial for the leading term; the important point for the error analysis is the order $\zeta^{-3/2}$.

To reduce the exact local jumps to \eqref{Airy-jumps}, continue $\hat r$ and $\tilde f$ analytically into the sectors of $B_\rho^H$ and set
\begin{equation}\label{cH}
 c_H(k;t):=2i\hat r(k)\tilde f(k)^2e^{2it\tilde g(H)}.
\end{equation}
Then the triangular local jumps occurring in RHP~\ref{RHP7} have the form
\begin{equation}\label{local-H-jumps-normalized}
 I+c_H(k;t)e^{2t\phi_H(k)}E_{21},
 \qquad
 I+c_H(k;t)^{-1}e^{-2t\phi_H(k)}E_{12},
\end{equation}
with the appropriate choice in each sector.  Therefore conjugation by $c_H(k;t)^{-\sigma_3/2}$, followed by the exponential factor in \eqref{Airy-phase-id}, converts the local problem exactly to the Airy model.

Define the analytic prefactor
\begin{equation}\label{EH}
 E_H(k):=\tilde S^\infty(k)
 c_H(k;t)^{-\frac12\sigma_3}M_{\rm Ai}^{-1}
 \zeta_H(k)^{\frac14\sigma_3},
 \qquad k\in B_\rho^H.
\end{equation}
Although the individual factors in \eqref{EH} have fourth-root behavior at $H$, their product is analytic and invertible in $B_\rho^H$; this follows from the identical fourth-root singularity of the outer parametrix $\tilde S^\infty$ at $H$.  The Airy local parametrix is then
\begin{equation}\label{PH-Airy}
 \tilde P^H(k)=E_H(k)\Psi_{\rm Ai}(\zeta_H(k))
 e^{t\phi_H(k)\sigma_3}c_H(k;t)^{\frac12\sigma_3},
 \qquad k\in B_\rho^H.
\end{equation}
By construction, $\tilde P^H$ has exactly the same jumps as $\tilde S$ inside $B_\rho^H$ and has the same local endpoint behavior at $H$.

On $\partial B_\rho^H$, one has $|\zeta_H(k)|\asymp t^{2/3}$ uniformly.  Substituting \eqref{Airy-asymptotics} into \eqref{PH-Airy} and using \eqref{EH} gives
\begin{equation}\label{PH-matching}
 \tilde P^H(k)(\tilde S^\infty(k))^{-1}
 =I+\mathcal O(t^{-1}),
 \qquad k\in\partial B_\rho^H,
\end{equation}
uniformly for $\xi$ in compact subsets of $(\hat\xi,-E_1)$.  The conjugate Airy parametrix is obtained from the Schwarz symmetry,
\begin{equation}\label{PHstar-Airy}
 \tilde P^{H^*}(k)=\sigma_2\bigl(\tilde P^H(k^*)\bigr)^*\sigma_2,
 \qquad k\in B_\rho^{H^*},
\end{equation}
and satisfies the same matching estimate.

For completeness, the remaining endpoints $E,E^*$ require modified-Bessel parametrices, not Airy parametrices.  Their construction is the one given in Section~\ref{TtoS}, with $g,f,S^\infty$ and the large parameter $|x|$ replaced by $\tilde g,\tilde f,\tilde S^\infty$ and $t$, respectively.  In particular, one may use the local conformal variables
\begin{equation}\label{zetaE-modulated}
 \zeta_E(k)=\frac14\bigl[it(\tilde g(k)-\tilde g(E))\bigr]^2,
 \qquad
 \zeta_{E^*}(k)=\frac14\bigl[it(\tilde g(k)-\tilde g(E^*))\bigr]^2,
\end{equation}
which are conformal at $E$ and $E^*$ because of the square-root behavior in \eqref{tg-local-behavior}.  The resulting local parametrices $\tilde P^E,\tilde P^{E^*}$ satisfy
\begin{equation}\label{PE-matching}
 \tilde P^\alpha(k)(\tilde S^\infty(k))^{-1}
 =I+\mathcal O(t^{-1}),
 \qquad k\in\partial B_\rho^\alpha,
 \qquad \alpha=E,E^*.
\end{equation}

We now define the global approximation
\begin{equation}\label{P-global-modulated}
 \tilde P(k)=
 \begin{cases}
 \tilde P^\alpha(k),&k\in B_\rho^\alpha,\quad
 \alpha\in\{E,E^*,H,H^*\},\\
 \tilde S^\infty(k),&\text{elsewhere},
 \end{cases}
\end{equation}
and introduce the error matrix
\begin{equation}\label{R-error-modulated}
 \tilde{\mathcal E}(k):=\tilde S(k)\tilde P(k)^{-1}.
\end{equation}
Its jump contour consists of the four circles $\partial B_\rho^\alpha$ together with the portions of the deformed contour outside the disks.  By \eqref{PH-matching} and \eqref{PE-matching},
\begin{equation}\label{error-circles}
 V_{\tilde{\mathcal E}}(k)=I+\mathcal O(t^{-1})
 \quad\text{on }\bigcup_\alpha\partial B_\rho^\alpha,
\end{equation}
while \eqref{teq1}--\eqref{teq4} imply
\begin{equation}\label{error-exp}
 V_{\tilde{\mathcal E}}(k)=I+\mathcal O(e^{-ct})
\end{equation}
on the remaining contour, for some $c>0$, uniformly for $\xi$ in compact subsets of $(\hat\xi,-E_1)$.  Hence
\begin{equation}\label{error-norm}
 \|V_{\tilde{\mathcal E}}-I\|_{L^2\cap L^\infty}
 =\mathcal O(t^{-1}).
\end{equation}
The standard small-norm RH theorem therefore yields
\begin{equation}\label{error-solution}
 \tilde{\mathcal E}(k)
 =I+\mathcal O\left(\frac1{t(1+|k|)}\right),
 \qquad t\to+\infty,
\end{equation}
uniformly for $k$ away from the error contour and for $\xi$ in compact subsets of $(\hat\xi,-E_1)$.  Consequently,
\begin{equation}\label{S-modulated-final}
 \tilde S(k;x,t)=
 \left(I+\mathcal O\left(\frac1t\right)\right)
 \tilde S^\infty(k;x,t),
\end{equation}
for $k$ away from the four endpoint disks.  This completes the local and error analysis required for the reconstruction formula below.

Based on these transformations $X(k)\to\tilde{X}(k)\to\tilde{T}(k)\to \tilde{S}(k)$, we can reconstruct the soliton gas $q(x,t)$.  For $k\to \infty$, the original RH problem for ${X}(k)$ is satisfied by
\begin{equation}\label{Eq8}
\begin{array}{l}
 X(k)=\tilde{X}(k) = \tilde{f}^{\sigma_3}_{\infty}e^{it\tilde{g}_{\infty}\sigma_3}\tilde{T}(k) e^{-it(\tilde{g}(k)-\theta(k))\sigma_3} \tilde{f}^{-\sigma_3}(k) = \tilde{f}^{\sigma_3}_{\infty}e^{it\tilde{g}_{\infty}\sigma_3}\tilde{S}(k) e^{-it(\tilde{g}(k)-\theta(k))\sigma_3} \tilde{f}^{-\sigma_3}(k)\\
  =\tilde{f}^{\sigma_3}_{\infty}e^{it\tilde{g}_{\infty}\sigma_3}\tilde{S}^{\infty}(k)\le( I+\mathcal{O}(\frac{1}{t})\ri) e^{-it(\tilde{g}(k)-\theta(k))\sigma_3} \tilde{f}^{-\sigma_3}(k).
\end{array}
\end{equation}
Recall the equation \eqref{u-X}, we obtain
\begin{align}\label{qfromSt}
q(x,t)=2i\tilde{f}^2_{\infty}e^{2it\tilde{g}_{\infty}}\le(\lim\limits_{k\to\infty}k\tilde{S}^{\infty}_{12}(k;x,t)\ri)+\mathcal{O}\le(\dfrac{1}{t}\ri).
\end{align}
Substituting \eqref{TildeSinfty_sol}  into \eqref{qfromSt}, we obtain \eqref{qxt1}.

\section{The unmodulated elliptic wave region: $\xi<\hat{\xi}$}
\label{sec:6}
Now let's consider the sector $\xi<\hat{\xi}$ to this end, we introduce a new $g$-function to analyse the RH problem for $X(k)$.
We need the function $\hat{g}(k)$ satisfies the following properties:
\begin{enumerate}
\item Jump condition:
\begin{align}
&\hat{g}_+(k) + \hat{g}_-(k) = 0,  & k \in {\Sigma}_1\cup {\Sigma}_2\label{hatgjump1},\\
&\hat{g}_+(k)-\hat{g}_-(k) =\hat{\Omega}, & k \in  \Sigma_{F}.\label{hatgjump2}
\end{align}
where the jump contours ${\Sigma}_1$, ${\Sigma}_2$ and $\Sigma_{F}$ see Fig.\ref{openinglenses1}.
\item Asymptotics at infinity:
\begin{align}
\hat{g}(k)-2k\xi-k^2=\hat{g}_{\infty}+\mathcal{O}\Big(k^{-1}\Big),\quad k\to\infty.
\end{align}
\end{enumerate}
In order to get an expression for $\hat{g}(k)$, we first define a new function $\check{g}(k;\hat{\xi})$
\begin{align}
\check{g}(k;\hat{\xi})=2\int\limits_{E}^{k}(\zeta-\mu(\hat{\xi}))\sqrt{\frac{(\zeta-F)(\zeta-F^*)}{(\zeta-E)(\zeta-E^*)}} \d\zeta,
\end{align}
which satisfies
\begin{enumerate}
\item Jump condition:
\begin{align}
&\check{g}_{+}(k) + \check{g}_{-}(k) = 0,  & k \in {\Sigma}_{1}\cup {\Sigma}_{2}, \\
&\check{g}_{+}(k)-\check{g}_{-}(k) =  \check{\Omega}, & k \in  \Sigma_{F},
\end{align}
where $\check{\Omega}=4\int\limits_{E}^{F}(\zeta-\mu(\hat{\xi}))\le(\sqrt{\frac{(\zeta-F)(\zeta-F^*)}{(\zeta-E)(\zeta-E^*)}}\ri)_{+} \d\zeta\in \R.$
\item Asymptotics at infinity
\begin{align}
\check{g}(k)=2k\hat{\xi}+k^2+\check{g}_{\infty}+\mathcal{O}\Big(k^{-1}\Big),\quad k\to\infty,
\end{align}
where
\begin{align}
\check{g}_{\infty}=(\int\limits_{E}^{\infty}+\int\limits_{E^*}^{\infty})[(\zeta-\mu(\hat{\xi}))\sqrt{\frac{{(\zeta-F)(\zeta-F^*)}}{{(\zeta-E)(\zeta-E^*)}}}-\hat{\xi}-\zeta]\d\zeta+E^2_2-E^2_1-2E_1\hat{\xi}\in \R.
\end{align}
\end{enumerate}
Utilizing the functions $\check{g}(k)$ and $g(k)$, where $g(k)$ is given by \eqref{gk}, we define $\hat{g}(k)$ explicitly as follows:
\begin{align}
\hat{g}(k)=\check{g}(k)+2(\xi-\hat{\xi})g(k),
\end{align}
then we have
\begin{align}
&\hat{\Omega}=\check{\Omega}+2(\xi-\hat{\xi})\Omega,\label{hw}\\
&\hat{g}_{\infty}=\check{g}_{\infty}+2(\xi-\hat{\xi})g_{\infty},\label{hginfty}
\end{align}
where $\Omega$ and $g_{\infty}$ are given by \eqref{Omega} and \eqref{ginfty}, respectively.
Note that near the endpoints $E,E^*,F,F^*$, the function $\hat{g}(k)$ has a square-root vanishing behaviour,
\begin{equation}\label{Eq1B}
\begin{array}{l}
\hat{g}_+(k) -\hat{g}_-(k) = \mathcal{O}\le( \sqrt{k-\alpha}\ri),\quad k \to \alpha,\ \alpha=E,E^*\\
\hat{g}_+(k) -\hat{g}_-(k) - \hat{\Omega} = \mathcal{O}\le( \sqrt{k-\beta}\ri),\quad k \to \beta,\ \beta=F,F^*.
\end{array}
\end{equation}
Combining \eqref{in1}, \eqref{in2}, \eqref{teq1} and \eqref{teq2}, we arrive at
\begin{align}
&\Re{[2i\check{g}(k)]}>0 ,\ \ \Re{[2ig(k)]}<0 \mbox{ for }k \in \mathcal{C}_{1} \backslash \{ E,F\},\\
&\Re{[2i\check{g}(k)]}<0 ,\ \ \Re{[2ig(k)]}>0 \mbox{ for }k \in \mathcal{C}_{2} \backslash \{ E^*,F^*\}.
\end{align}
where the contours $\mathcal{C}_{1}$ and $\mathcal{C}_{2}$ are shown in Fig.\ref{openinglenses1}.
Therefore, as $\xi<\hat{\xi}$, we have
\begin{align}
&\Re [2 i\hat{g}(k)]=\Re{[2i\check{g}(k)]}+2(\xi-\hat{\xi})\Re{[2ig(k)]}>0 \ \mbox{ for }k \in \mathcal{C}_{1} \backslash \{E,F \},\\
&\Re [2 i\hat{g}(k)]=\Re{[2i\check{g}(k)]}+2(\xi-\hat{\xi})\Re{[2ig(k)]}<0 \ \mbox{ for }k \in \mathcal{C}_{2} \backslash \{E^*,F^* \}.
\end{align}
We now record the RH reduction in this region explicitly.  Define
\begin{equation}\label{hatT-def}
 \hat T(k)=f_\infty^{-\sigma_3}e^{-it\hat g_\infty\sigma_3}
 X(k)e^{it(\hat g(k)-\theta(k))\sigma_3}f(k)^{\sigma_3},
\end{equation}
where the scalar function $f$ is the one introduced in Section~\ref{TtoS}.  Using \eqref{hatgjump1}--\eqref{hatgjump2}, the jumps on $\Sigma_1\cup\Sigma_2$ are factorized exactly as in the large-negative-$x$ analysis, with $xg$ replaced by $t\hat g$.  Opening the same lenses $\mathcal C_1,\mathcal C_2$ defines a matrix $\hat S$ whose jumps on the lips contain $e^{\pm2it\hat g(k)}$.  The two inequalities above imply that these jumps converge exponentially to the identity away from the endpoints.

The resulting outer model has the constant jumps
\begin{equation}\label{hat-outer-jumps}
 \hat S^\infty_+(k)=\hat S^\infty_-(k)
 \begin{cases}
 \begin{pmatrix}0&i\\ i&0\end{pmatrix},&k\in\Sigma_1\cup\Sigma_2,\\[2mm]
 \begin{pmatrix}e^{i(t\hat\Omega+\Delta)}&0\\0&e^{-i(t\hat\Omega+\Delta)}\end{pmatrix},&k\in\Sigma_F,
 \end{cases}
\end{equation}
with $\hat S^\infty(k)=I+\mathcal O(k^{-1})$ as $k\to\infty$.  Hence $\hat S^\infty$ is given by the theta-function formula \eqref{Sinfty_sol} after the single replacement
\begin{equation}\label{hat-outer-replacement}
 x\Omega+\Delta\longmapsto t\hat\Omega+\Delta.
\end{equation}

All four endpoints are hard edges in the present region because of \eqref{Eq1B}; therefore the local models are modified-Bessel parametrices.  For each $\alpha\in\{E,E^*,F,F^*\}$ choose a disk $B_\rho^\alpha$ and define
\begin{equation}\label{hat-zeta-alpha}
 \zeta_\alpha(k)=\frac14\left[it\bigl(\hat g(k)-\hat g(\alpha)\bigr)\right]^2.
\end{equation}
The square-root behavior in \eqref{Eq1B} implies that $\zeta_\alpha$ is conformal at $\alpha$.  Using the same modified-Bessel model $\Psi_{\rm Bes}$ as in Section~\ref{TtoS}, and choosing the analytic prefactor by matching to $\hat S^\infty$, we obtain local parametrices $\hat P^\alpha$ satisfying
\begin{equation}\label{hat-bessel-matching}
 \hat P^\alpha(k)(\hat S^\infty(k))^{-1}
 =I+\mathcal O(t^{-1}),
 \qquad k\in\partial B_\rho^\alpha,
 \qquad \alpha=E,E^*,F,F^*.
\end{equation}
Define the global approximation $\hat P$ by $\hat P=\hat P^\alpha$ in the four disks and $\hat P=\hat S^\infty$ elsewhere, and set
\begin{equation}\label{hat-error}
 \hat{\mathcal E}(k)=\hat S(k)\hat P(k)^{-1}.
\end{equation}
Then the circle jumps satisfy $V_{\hat{\mathcal E}}=I+\mathcal O(t^{-1})$, while all remaining jumps satisfy $V_{\hat{\mathcal E}}=I+\mathcal O(e^{-ct})$.  Consequently,
\begin{equation}\label{hat-error-est}
 \hat{\mathcal E}(k)=I+\mathcal O\left(\frac1{t(1+|k|)}\right),
 \qquad t\to+\infty,
\end{equation}
uniformly for $\xi$ in compact subsets of $(-\infty,\hat\xi)$.  Reversing the transformations and using \eqref{u-X} gives \eqref{qxt2}.

\noindent{\bf Acknowledgments}
This work is supported by the National Natural Science Foundation of China (Grant Nos. 12471234 and 12271490), Science Foundation of Henan Academy of Sciences (Grant No. 20252319002), the National Natural Science Foundation of Henan Province (Grant No. 242300421158) and the Outstanding Young Talents in Scientific and Technological Innovation of Zhongyuan.

\noindent{\bf Data availability} We do not analyze or generate any datasets, because our work proceeds within a theoretical and mathematical approach.

\section*{Declarations}
\noindent{\bf Conflict of interest} The authors declare that there are no conflict of interest.

\end{document}